\documentclass[lettersize,journal]{IEEEtran}
\usepackage[draft]{changes}
	\definechangesauthor[name={Tao Li}, color=orange]{TL}
\newcommand{\tp}{\mathsf{T}}
\usepackage{cite}
\usepackage{enumitem}
\usepackage{url}            
\usepackage{booktabs}       
\usepackage{amsfonts}       
\usepackage{nicefrac}       
\usepackage{microtype}      
\usepackage{lipsum}
\usepackage{amsthm}
\usepackage{amssymb}
\usepackage{amsmath}
\usepackage{amsfonts}
\usepackage{hyperref}       
\usepackage{cleveref}
\newtheorem{theorem}{Theorem}

\newtheorem{corollary}{Corollary}
\newtheorem{proposition}{Proposition}

\newtheorem{definition}{Definition}
\newtheorem{remark}{Remark}
\newtheorem{assumption}{Assumption}
\usepackage{xcolor}

\usepackage{mathtools}

\usepackage[english]{babel}

\usepackage{comment}

\usepackage{algorithm}
\usepackage[noend]{algpseudocode}
\makeatletter
\def\BState{\State\hskip-\ALG@thistlm}
\makeatother

\usepackage{subfigure}
\DeclareMathOperator*{\argmin}{arg\,min}

\begin{document}

\title{Distributed Machine Learning with Strategic Network Design: A Game-Theoretic Perspective}

\author{Shutian~Liu,~\IEEEmembership{Graduate Student Member,~IEEE},
Tao Li,~\IEEEmembership{Graduate Student Member,~IEEE}
        and~Quanyan~Zhu,~\IEEEmembership{Member,~IEEE}
\thanks{The authors are with the Department of Electrical and Computer Engineering, Tandon School of Engineering,
        New York University, Brooklyn, NY, 11201 USA (e-mail: \{sl6803, tl2636, qz494\}@nyu.edu).}       
       
}



\maketitle

\begin{abstract}
This paper considers a game-theoretic framework for distributed machine learning problems over networks where the information acquisition at a node is modeled as a rational choice of a player.
In the proposed game, players decide both the learning parameters and the network structure.
The Nash equilibrium characterizes the tradeoff between the local performance and the global agreement of the learned classifiers.
We first introduce a commutative approach which features a joint learning process that integrates the iterative learning at each node and the network formation. We show that our game is equivalent to a generalized potential game in the setting of undirected networks. We study the convergence of the proposed commutative algorithm, analyze the network structures determined by our game, and show the improvement of the social welfare in comparison with standard distributed learning over fixed networks.
To adapt our framework to streaming data, we derive a distributed Kalman filter.
A concurrent algorithm based on the online mirror descent algorithm is also introduced for solving for Nash equilibria in a holistic manner.
In the case study, we use telemonitoring of Parkinson's disease to corroborate the results. 
\end{abstract}


\section{Introduction}
\label{sec:intro}
%
%
%
%

\IEEEPARstart{D}{istributed} machine learning has been widely used to handle large-scale machine learning tasks \cite{peteiro2013survey}. It provides a mechanism for naturally distributed data sources in large-scale learning problems. For example, autonomous vehicles collect spatial data in an urban environment and form a V2V communication network to share the data for making better decisions.   In the Internet of Things (IoT) systems, the owners of the devices can share privately-owned security information to reduce their cyber risks collaboratively.  The communication and sharing of information among nodes in the network enables nodes with limited computational resources to improve their learning capabilities through a collaborative mechanism.

The literature of distributed learning has focused largely on the scenario where the goal is to reach a consensus of learning parameters given fixed networks.
However, the setting of fixed networks may restrict the applications of distributed learning schemes.
One example where a fixed network is insufficient is federated learning \cite{mcmahan2017communication}.
In federated learning problems, one aims to design self-ruling agents who decide on their own when to join the distributed learning problem.
Hence, assuming a fixed network for communicating learning parameters or gradient updates violates one of the innovations of federated learning.
Furthermore, federated learning often takes into account the instabilities of communication links between the nodes. 
This instability may arise as a result of a mobile device pausing its local learning process to save power or because of the interference on the wireless communication channels.
Fixed networks certainly cannot resolve the challenge of unstable communications.
Another example is distributed learning with biased data sets. 
When the local data samples of the learning agents follow different distributions, the consensus of learning parameters is not the only target one can aim for.
Instead, one may consider distributed learning with multiple consensuses where subgroups of learning agents agree on different consensuses suitable for their local data distributions.
In this scenario, one needs to figure out the network structure after obtaining the learning parameters of agents.
The fixed network assumption is not feasible in this situation.
Therefore, there is a need to develop new distributed learning paradigms to explore richer network structures.

\begin{figure}[ht]
\centering
\vspace{-5mm}
\includegraphics[width=0.5\textwidth]{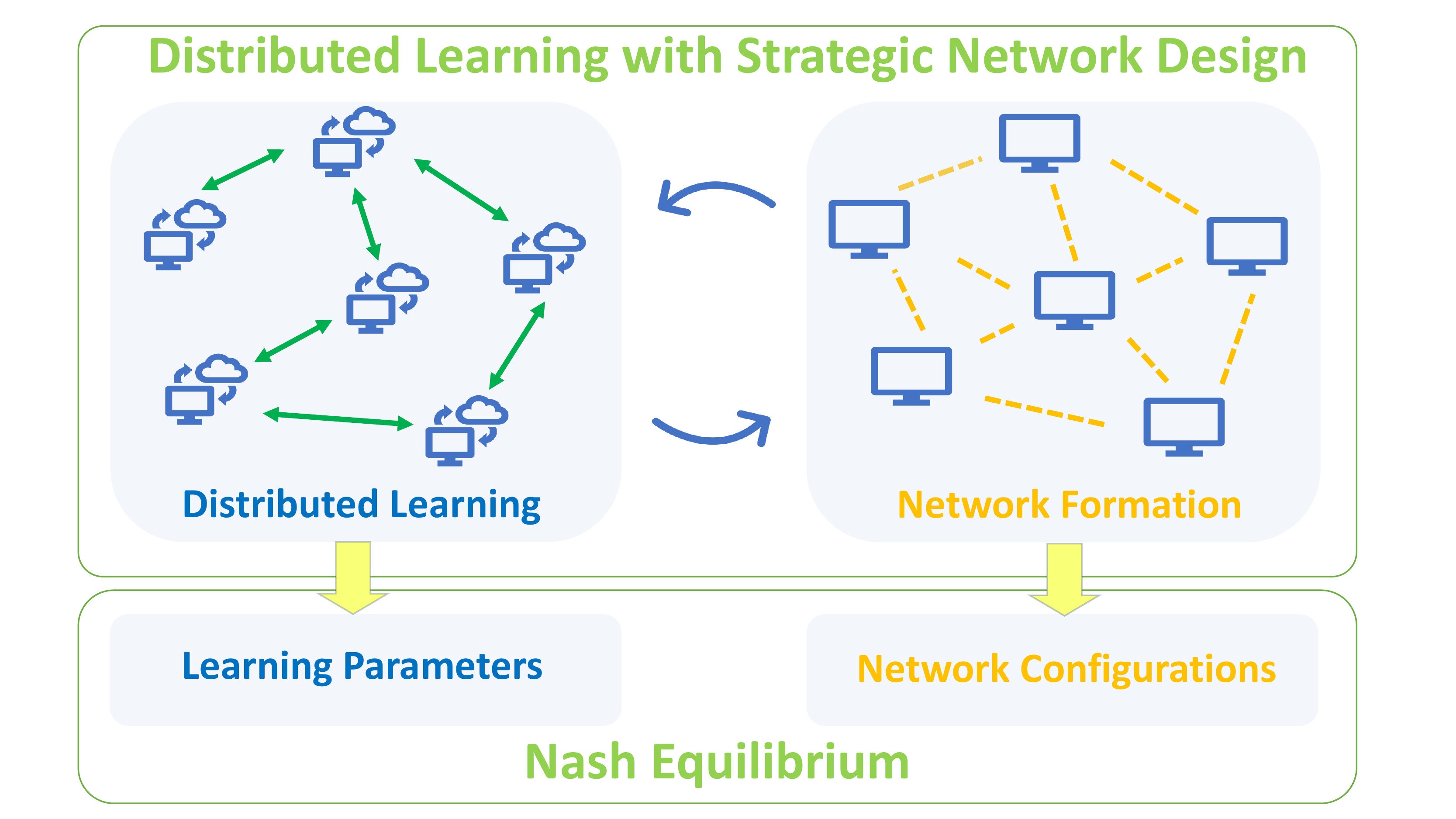}
\vspace{-7mm}\caption[Optional caption for list of figures]{The distributed learning with strategic network design framework. The distributed learning layer conducts distributed machine learning over a fixed network. The network formation layer computes the most efficient network given the learning parameters. Both commutative and concurrent learning methods lead to Nash equilibria.
} 
\label{fig:conceptual}
\end{figure}

In this paper, we introduce a game-theoretic framework for distributed learning with strategic network design.
In the proposed game, we model each learning agent (a node on the network) as a rational player with two distinct actions, as depicted in Fig. \ref{fig:conceptual}.
One action is deciding the optimal learning parameters to yield the minimum learning error based on the local data. The other action is to choose the weights of the links connecting this node to its neighbors. A positive link weight indicates that one player is willing to connect and that she would set her learning parameter close to that of the connected player.
The two actions are interconnected by the design of the players' disutility functions.
The Nash equilibrium of the game requires that no player has incentive to adjust her learning parameter or reconfigure her connections with other players.

Our framework has the following features.
Firstly, the network structure is an outcome of the proposed game.
Since the players in our game decide the link weights by themselves, our framework enriches the literature of distributed learning by enabling agent-configured local connectivities.
This property aids one to figure out the network structure at the Nash equilibrium, which is suitable for the scenarios of federated learning or biased datasets.
Furthermore, this property can be interpreted as the efforts of the players in finding the optimal networked information structure for playing the game.
Secondly, by the design of the players' disutility functions, we can obtain multiple soft consensuses of the learning parameters. 
The multiplicity property makes our framework adaptable to large-scale problems where a single consensus is insufficient to capture the whole learning problem.
By soft consensuses, we refer to the fact that we transform the equality consensus constraints in standard distributed learning problems to players' costs that punish the misalignment of learning parameters. 
This softening not only enables the players to find partners to form consensuses but also aids the computation of equilibria.

To break the coupling of the two actions, we introduce a commutative approach that iterates between the local learning and the global network formation.
In the first layer, players best-respond to the other players learning actions given the current network configuration. 
In the second layer, players refine the network structure based on the state-of-the-art learning parameters.
To support the proposed commutative approach, we prove that our game belongs to the class of weighted potential games if we restrict to undirected networks.
This property leads to the convergence of the algorithm associated with the commutative approach.
We analyze the network structures using an approach inspired by the cohesiveness defined in \cite{morris2000contagion}.
From an optimization perspective, we show that the our game-theoretic framework admits better social welfare compared to a standard distributed learning framework.
Apart from the commutative approach, we also provide a concurrent method to find the NE for a broader class of network structures based on the online mirror descent algorithm.
A distributed Kalman filter is also derived for processing streaming data. 

Finally, our results are corroborated in a case study using data of telemonotoring measurements of Parkinson's disease. 
We further investigate the effects of reference information by comparing the local learning performances at a node when it connects to and disconnects from other nodes.

This paper is organized as follows. Section \ref{sec:related} reviews the related works. Section \ref{sec:framework} presents the proposed game-theoretic framework. In Section \ref{sec:analysis}, we first show the existence of the NE, and then we introduce the commutative approach to finding the NE. In Section \ref{sec:analysis}, we focus on the properties of our framework under undirected networks. We present convergence analysis of the commutative algorithm. We analyze the potential outcomes of network structures and compare the social welfare obtained using our framework and the one obtained using standard distributed learning frameworks.
A concurrent method for equilibrium seeking in the generic setting is presented in Section \ref{sec:single loop}.
We devote Section \ref{sec:streaming} to the adaptation of our framework to streaming data.
Section \ref{sec:case} uses a case study to corroborate our results. 
Section \ref{sec:conclusion} concludes the paper.

\section{Related Works}
\label{sec:related}

Our framework builds on the vast literature on distributed optimization over networks, which lays the foundation of distributed learning.
We refer to the survey papers \cite{peteiro2013survey,yang2019survey} and the references therein for a comprehensive review of distributed optimization over networks and its connection with distributed machine learning.
While most of the existing works focus on the setting of fixed networks, our approach considers an additional optimization procedure to find the optimal network configuration.
This idea is closely related to the scenarios of time-varying networks considered, for instance, in \cite{nedic2014distributed,vyavahare2019distributed,xu2017dynamic}.
The difference lies in that, in our framework, the dynamic changing nature of the network is the result of seeking the optimal network configuration rather than the consequence of a given network dynamics.
Decentralized algorithms also play an important role in distributed learning. We refer the readers to the monograph \cite{boyd2011distributed} and the references therein for algorithms based on the method of multipliers.
Note that convexity and differentiability conditions are essential to the analysis of dual-based algorithms. 
We refer to \cite{uribe2020dual} for the study of convergences of dual-based algorithms under a variety of assumptions.

There is a recent trend in using potential games to model distributed learning problems \cite{van2017distributed,ali2019distributed, marden2012state}. 
One of the advantages of potential games over other types of games lies in that the players' objective functions in a potential game can be described by a potential function which represents the joint target of all the players.
This holistic representation of the players' incentives leads to the conveniences in computing the Nash equilibrium (NE).
In \cite{li2013designing}, the authors have introduced the approach of designing agents' local objective functions to reach a specific system-level equilibrium using state-based potential games proposed in \cite{marden2012state}. 
The formulation of players' disutility functions in our framework is inspired by \cite{li2013designing}.
Furthermore, we equip the players in the game beyond merely the incentives to optimize their local learning parameters.
In particular, the players have incentives to look for the most efficient information structure for their learning tasks.
This additional incentive leads to network configurations at equilibria. 

The proposed learning schemes for equilibrium seeking in this work fall within the realm of decentralized game-theoretic learning \cite{tao-confluence}. Each player updates her strategy at each iteration using only local information, such as neighbors' actions and messages. When dealing with undirected networks, the proposed game problem admits a weighted potential. Motivated by the encouraging results in \cite{dubey2006strategic}, we adopt the best response dynamics in the first layer equilibrium-seeking and prove its convergence to generalized Nash equilibrium. However, the convergence guarantee no longer holds for directed networks, due to asymmetric network structures. Inspired by recent advances on gradient-based learning dynamics \cite{tao-confluence,bravo2018bandit,mertikopoulos2019learning}, the second learning scheme based on online mirror descent is proposed to address directed network structures. Unlike the best-response-based two layer approach, online mirror descent algorithm enables the player to update the learning action and the network formation simultaneously, leading to a concurrent decentralized learning scheme.

The literature on network formation from both engineering and economics are also relevant to us.
Economists have investigated network formation problems using games, for example, in \cite{bramoulle2007public,galeotti2010law}.
These classic works have established analysis of the network patterns resulting from resource allocation and information distribution.
Our work is related to \cite{chen2019interdependent}, where the authors have utilized the games-in-games framework \cite{zhu2015game} to study the security issues in the Internet of Things. 
Efficient network configurations are consequences of the limited cognitive attentions players pay to the others.
In our approach, the pursuit of an efficient network structure comes out of the players' incentives to improve the performance of local learning.
Such an improvement relies on gathering other players' learning parameters for reference.

\section{Game-theoretic Framework}
\label{sec:framework}
In this section, we introduce our game-theoretic framework for distributed machine learning with strategic network design.
In our network setting, we consider a graph with the set of nodes $\mathcal{N}:=\{1,2,...,N\}$. We use $i$ and $j$ to denote typical elements in $\mathcal{N}$. Each node $i\in\mathcal{N}$ possesses local data $\{x_{i,k},y_{i,k}\}_{k=1}^{K_i}$ to be used in the distributed learning task, where $x_{i,k}\in\mathbb{R}^d$, $y_{i,k}\in\{1,-1\}$ and $K_i\in\mathbb{N}$. We can interpret $x_{i,k}$ as the features obtained from pre-processing of raw data and $y_{i,k}$ as the labels. The integer $K_i$ is the total number of available data points at node $i$. In this paper, we use the terms ``node" and ``player" interchangeably. 

\par
We use $m^i=[m^i_j]_{j\neq i,j\in\mathcal{N}}\in[0,1]^{N-1}, \forall i\in\mathcal{N}$ to denote the weights on the directed links from node $i$ to other nodes in $\mathcal{N}$. 
The interpretation of $m^i_j$ can be the willingness of node $i$ in cooperating with node $j$ in the learning task. A positive $m^i_j$ indicates that observing learning information from node $j$ is beneficial for node $i$. While when $m^i_j=0$, node $i$ is not interested in communicating with node $j$, or exchanging information with node $j$ has such a high communication cost that she would prefer maintaining local. We can also interpret $m^i_j$ as the amount of attention node $i$ paid to node $j$. We remark that negative $m^i_j$ is also possible, especially when we consider the scenario where nodes are set to participate in the learning task. When a node has to communicate with other nodes, a negative link weight captures her loss. For the purpose of presentation, we consider only positive link weights.
\par
The actions and the cost functions of players illustrate the structure of the game-theoretic framework. One feature of our framework is that a player has two distinct but interrelated actions. The two actions are associated with the players' learning using local data and their communications with others to exchange learning information. 
\par
Player $i$'s first action, also called the learning action, $u_i\in \mathcal{U}_i \subset \mathbb{R}^d$ bridges her local data and her learning cost. We assume that $\mathcal{U}_i,i\in\mathcal{N}$ are convex and compact. This action can take the form of the weights or parameters of the learning problem. Thus, $x_{i,k}^Tu_i \in \mathbb{R}$ represents player $i$'s label prediction obtained using her learning action $u_i$ and her local data point with index $k$. The $k^{th}$ prediction error is given by $y_{i,k}-x_{i,k}^Tu_i$. The reason why we call $u_i$ an action instead of a weight vector is that $u_{j},j\in\mathcal{N}$ of other players in the game influence the choice of $u_i$. We also refer to $u_i,i\in\mathcal{N},$ as a classifier. 
\par
Player $i$'s second action is the network formation decision $m^i\in\mathcal{M}_i:= [0,1]^{N-1}$ which represents the link weights. By making the link weights actions of the players, we can analyze the network structures formed by the players themselves. In this case, the rationality of player $i$ leads to a choice of $m^i$ in a way that improves her learning result through exchanging information with other players at the cost of communication.
To avoid trivial solutions to our problem, we assume that all players in $\mathcal{N}$ are willing to join the distributed learning task and node $i\in\mathcal{N}$ has a presumed budget $\beta_i\in\mathbb{R}^{++}$ to consume on communicating with other players, i.e., $\mathrm{1}^Tm^i=\beta_i$. 
Hence, we consider the set of network formation actions given by $\mathcal{M}_i=\{m^i:m^i\in[0,1]^{N-1}, \mathrm{1}^Tm^i=\beta_i\}$. 
Note that $\beta_i$ also has the interpretation of the limited attention of player $i$.

\par
The strategy profile of players is $s:=[u, m]=[u_1,...,u_N,m^1,...,m^N]\in \mathcal{S}:= \mathcal{U}\times \mathcal{M}$ with $\mathcal{U}:=\mathcal{U}_1\times \cdots \times \mathcal{U}_N \subset \mathbb{R}^{d\times N}$ and $\mathcal{M}:=\mathcal{M}_1\times\cdots\times\mathcal{M}_N\subset\mathbb{R}^{(N-1)\times N}$. Let $s_i=[u_i,m^i]\in\mathcal{S}_i=\mathcal{U}_i\times\mathcal{M}_i$.
\par

Define the cost function $J^i:\mathcal{S}\rightarrow \mathbb{R}$ of player $i$ as
\begin{equation}
    J^i(u_i,u_{-i},m^i,m^{-i})=\alpha_i l^i(u_i)+\sum_{j\neq i, j\in\mathcal{N}}m^i_j||u_i-u_j||^2_2,
    \label{eqn general cost func}
\end{equation}
where $u_{-i}$ and $m^{-i}$ denote the learning actions and network formation actions of players other than player $i$, respectively.
The mapping $l^i$ us defined as $l^i:\mathcal{U}_i\rightarrow \mathbb{R}$. 
The parameter $\alpha_i\in \mathbb{R}^+$ balances the two sources of costs in $J^i$.

The first term $l^i$ in (\ref{eqn general cost func}) captures the empirical cost of node $i$ induced by the learning task using local data. 
We assume that $l^i$ has already captured the local differences between the nodes, such as the choices of learning methods and the distributions of the local data.
Moreover, even the learning problems themselves may differ for nodes. This consideration is practical, since nodes in a network are distinct in various aspects. 
In practice, we can choose $l^i$ as the sum of a proper loss function of the prediction errors $y_{i,k}-x_{i,k}^Tu_i, k=1,...,K_i$. 
We assume that $l^i, i\in \mathcal{N},$ are continuously differentiable and convex.
We will discuss the scenarios where the loss functions are nonsmooth and nonconvex in Section \ref{sec:single loop}.

The second term $\sum_{j\neq i,j\in\mathcal{N}}m^i_j||u_i-u_j||^2_2$ in (\ref{eqn general cost func}) collects the influences from other players' learning actions through the links of player $i$ specified by $m^i$. 
This expression captures the disagreement of learning actions between player $i$ and other players in the game weighted by the term $m^i_j$. 
In contrast to a common distributed learning task where the weights of neighboring nodes $i$ and $j$ must satisfy the constraint $u_i=u_j$, we relax this hard constraint to a cost induced by the relative distance of $u_i$ and $u_j$. 
This relaxation makes the following considerations related to federated learning possible. 
First, it emphasizes the local learning performances instead of the learning efficiency of a centralized model learned by all the nodes cooperatively. 
Second, it captures the self-rule feature in the federated learning. 
In particular, if a node has little interest in joining the distributed learning task, we can set the weighting parameter $\alpha_i$ of this node to be relatively large, resulting in a negligible contribution from $\Bar{l}^i$ to $J^i$.

The cost of communications is one of the central challenges in distributed learning problems, especially in federated learning problems. 
We show in the following that there is an implicit communication cost included in the cost function $J^i$ due to the budget $\beta_i$ considered in the set $\mathcal{M}_i$.
Given $u$ and $m^{-i}$, the optimization problem for player $i$ to select her network formation action is 
\begin{equation}
    \min_{m^i\in\mathcal{M}_i} J^i(u_i,u_{-i},m^i,m^{-i}).
    \label{eq:player i's network formation problem}
\end{equation}
Consider the scenario where each communication link player $i$ builds induces a cost of $\theta_i\in\mathbb{R}^{++}$.
Then, $\theta_i||m^i||_0$ represents the total communication cost of player $i$ since $||m^i||_0$ represents the number of links.
We know from compressed sensing \cite{donoho2006compressed} that the solution to (\ref{eq:player i's network formation problem}) is sparse when we regularize the objective function $J^i$ with $\theta_i||m^i||_0$.
A commonly used approximation to the nonconvex regularization term $\theta_i||m^i||_0$ is its $l_1-$norm counterpart.
In other words, we obtain sparse networks by solving the following problem:
\begin{equation}
    \min_{m^i\in\mathcal{M}_i} J^i(u_i,u_{-i},m^i,m^{-i})+\theta_i||m^i||_1.
    \label{eq:equivalent player i's network formation problem}
\end{equation}
Furthermore, (\ref{eq:equivalent player i's network formation problem}) is equivalent to (\ref{eq:player i's network formation problem}), for $\theta_i||m^i||_1$ is a constant with respect to $m^i$ under the constraint $m^i\in\mathcal{M}_i$.
Therefore, the network configuration obtained by solving problem (\ref{eq:player i's network formation problem}) not only takes into account the communication costs to players, but also admits a sparse pattern.
The sparsity of the solution can aid the players avoid spending efforts in building unnecessary links with others.

We remark here the reasons why a self-ruled player may not have the incentive to join the distributed learning task. First, the player can possess an adequate amount of high-quality local data. This means that the player can fully rely on herself when there is a local learning task. From the perspective of a large-scale learning problem, data from remote sources may follow different distributions. This further results in little or even negative contribution to the player's local learning task whose goal is to serve local users. Second, cooperation with other nodes in the network can induce losses. The cooperation with a node who has an unstable connection or a relatively weak computation power will lead to a significant delay in the learning process, since parameter updates have to wait until all the gradient computations are finished at the nodes. Cooperation can also cause security issues. Even in the case where all the local data stay local, the exchanged weights or parameters can be used by malicious ones to infer the original data \cite{fredrikson2015model}. This can cause severe leaks of private information, especially when the learning is performed on medical data.

\par
Define $J=[J^1,...,J^N]$. We use the tuple $\mathcal{G}:=\langle \mathcal{N},\mathcal{S},J \rangle$ to denote the game described above. The following definition presents the solution concept of the game $\mathcal{G}$.
\begin{definition}
\label{def:NE}
[NE of $\mathcal{G}$]
A strategy profile $s^*=[u^*,m^*]=[u_1^*,...,u_N^*,m^{1*},...,m^{N*}]$ is an NE of the game defined by the tuple $\mathcal{G}$, if $\forall u_i\in \mathcal{U}_i$ and $\forall m^i\in \mathcal{M}_i$
\begin{equation}
    J^i(u_i^*,u_{-i}^*,m^{i*},m^{-i*})\leq J^i(u_i,u_{-i}^*,m^{i},m^{-i*}), \forall i\in\mathcal{N}.
    \label{eqn NE def of 3 term cost}
\end{equation}
\end{definition}
\par
In general, obtaining the NE of the game $\mathcal{G}$ is non-trivial, since the two actions of players are coupled. To overcome the challenge, we present a commutative decision-making procedure to decouple the problem.

\section{Game-Theoretic Analysis}
\label{sec:analysis}
In this section, we first show the existence of NE.
Then, we introduce the commutative approach to compute the equilibrium strategy profile of the game $\mathcal{G}$ by iteratively performing distributed learning over a fixed network and updating the network configuration under fixed learning parameters.

\subsection{Existence of NE}
\label{sec:analysis:existence}

\begin{assumption}
A player $i\in\mathcal{N}$ has a convex and compact learning action set $\mathcal{U}_i\subset \mathbb{R}^d$ and a network formation action set of $\mathcal{M}_i=\{m^i:m^i\in[0,1]^{N-1}, \mathrm{1}^Tm^i=\beta_i>0\}$. Her local learning cost function $l^i(\cdot)$ is continuously differentiable and convex. The scalars $\alpha_i\in \mathbb{R}^+$ and $\theta_i\in \mathbb{R}^+$.
\label{assump 1}
\end{assumption}
\begin{theorem}
[Existence of NE]
Under Assumption \ref{assump 1}, the $N$-player non-zero sum infinite game $\mathcal{G}$ admits an NE in pure strategies.
\label{coro existence of NE}
\end{theorem}
\begin{proof}
Under Assumption \ref{assump 1}, the action set $\mathcal{S}$ is compact and convex with respect to both $u$ and $m$. The cost functions for players in (\ref{eqn general cost func}) are continuous in $u_{-i}$, continuous and quasi-convex in $u_i$ and linear in $m^i$. It follows from Theorem 1.2 of \cite{fudenberg1991game} that there exists a pure-strategy NE $s^*$ in the game $\mathcal{G}$.
\end{proof}
Note that even when the cost functions $l^i(\cdot)$ are quasi-convex, the existence result still holds (See \cite{fudenberg1991game}).

\subsection{Commutative Approach}
\label{sec:analysis:two layer}
Equilibrium seeking is challenging in general. With the introduction of the actions $m^i$, our framework becomes more sophisticated due to the multiplicity of player's actions. To deal with the joint decision-making of determining both the learning parameters and the network configuration, we consider a commutative approach. In the first layer, we consider a fixed network. The players learn the equilibrium strategies $u_i^*,i\in\mathcal{N}$ given the link weights $m^i, i\in\mathcal{N}$. In the second layer, players observe the equilibrium strategies $u_i^*$, and further optimize their cost functions by selecting the link weights $m^i$. We solve $\mathcal{G}$ by iterating the algorithms at the two layers.
\par
\subsubsection{Learning under Fixed Networks}
In the first layer, we adopt a learning scheme for the players to determine the strategies at an NE under fixed network link weights $m$. 
\par
Consider the best-response dynamics. Player $i\in\mathcal{N}$ computes the best responses using (\ref{eqn general cost func}) given the other players learning action $u_{-i}$. The best response of player $i$ is obtained by solving
\begin{equation}
    u_i^*=BR^i(u_{-i},m^i):=\min_{u_i\in\mathcal{U}_i} J^i(u_i,u_{-i},m^i,m^{-i}).
    \label{eqn prob to solve for BR}
\end{equation}
When the cost of local learning $l^i(\cdot)$ takes a specific  form, such as linear or quadratic ones, we can derive the analytical solution of (\ref{eqn prob to solve for BR}) using first-order optimality conditions. In general, we solve the convex optimization problem (\ref{eqn prob to solve for BR}) numerically using standard optimization techniques, such as the interior point method \cite{boyd2004convex}. 
\par
We use best-response dynamics to learn the equilibrium for each player under fixed network. We let $u_{i,t}$ denote the action of player $i$ at time $t$. Given an initialization of actions $u_{i,0}, i\in\mathcal{N}$, under the fixed network $m$, players update their actions according to the following best-response dynamics
\begin{equation}
    u_{i,t+1}=BR^i(u_{-i,t},m^i),
    \label{eqn BRD}
\end{equation}
where $u_{-i,t}=\begin{pmatrix}
    u_{1,t}  & \dots & u_{i-1,t} & u_{i+1,t} & \dots
    & u_{N,t} \\
\end{pmatrix}$ are the actions of players other than player $i$ at time $t$. 
Best-response dynamics are often used for the computation of NE for continuous-kernel games because of its straightforward interpretation.
\par

\par
\subsubsection{Refinement of the Network Structure}
In the second layer, players proceed to obtain the most efficient link weights by further optimizing (\ref{eqn general cost func}) with respect to $m^i$. The network formation problem of player $i$ is
\begin{equation}
    \begin{aligned}
        \min_{m^i} \  \ &J^i(u_i^*,u^*_{-i},m^i) \\
        s.t. \  \ &m^i\in\mathcal{M}_i,
    \end{aligned}
    \label{eqn netowrk formation prob}
\end{equation}
where the learning actions are obtained from (\ref{eqn BRD}) or (\ref{eq:noise_omd}). 
Problem (\ref{eqn netowrk formation prob}) is also a convex optimization problem. 
Note that the solution of (\ref{eqn netowrk formation prob}) is sparse, since it is equivalent to optimizing $J^i(u_i^*,u^*_{-i},m^i)+\theta_i||m^i||_1$.

\subsubsection{Commutative Algorithms}
The commutative approach for solving $\mathcal{G}$ combines the computation of an NE given the network and the efficient link weights under the equilibrium strategies at the NE. We summarize the procedure in Algorithm \ref{alg:1}, whose convergence issues will be discussed later in Section \ref{sec:symmetric}.

\begin{algorithm}
\caption{Commutative Distributed Learning with Strategic Network Design.}\label{alg:1}
\begin{algorithmic}[1]
\State{Initialization of $u_{i}$ and $m^i$, $\forall i \in \mathcal{N}$}
\State Fix $m^i, i\in\mathcal{N}$, update $u$ using the best response dynamics in (\ref{eqn BRD}) for sufficiently many rounds.
\State Fix $u_i, i\in\mathcal{N},$ obtained in step $2$, compute optimal link weights using (\ref{eqn netowrk formation prob}) until convergence.
\State Repeat step $2$ and step $3$ 

\end{algorithmic}
\end{algorithm}

\section{Properties under Undirected Network}
\label{sec:symmetric}
In this section, we study the properties of game $\mathcal{G}$ under the scenario where the links between two nodes are undirected.
In an undirected network, the communication over a link is bi-directional. Furthermore, a link weight between two nodes characterizes the unique attention factor shared by both of the nodes paid to the other node.

\subsection{Global Problem}
\label{sec:symmetric:global problem}
We first state the condition that is essential in the analysis of this section. 
\begin{assumption}
The network specified by $m$ is undirected, i.e. $m^i_j=m^j_i, \forall i,j\in\mathcal{N}, i\neq j$.
\label{asspm 2 sysmetric network}
\end{assumption}
With a undirected network, we conclude the following result.
\par
\begin{theorem}
Under Assumption \ref{assump 1} and Assumption \ref{asspm 2 sysmetric network},
the game $\mathcal{G}$ is a generalized weighted potential game \cite{monderer1996potential} \cite{facchinei2011decomposition}. The potential function is given by
\begin{equation}
    \Phi(u,m) =\sum_{i\in\mathcal{N}}\alpha_il^i(u_i)
    +\frac{1}{2}\sum_{i\in\mathcal{N}}\sum_{j\in\mathcal{N},j\neq i}m_j^i||u_i-u_j||_2^2.
\label{eqn potential func general}
\end{equation}
The $i^{th}$ weight is given by $W_i=diag[\mathrm{1}_d^T \  \ 2\cdot\mathrm{1}_{N-1}^T]^T\in\mathbb{R}^{(d+N-1)\times (d+N-1)}, \forall i\in\mathcal{N}$, where $\mathrm{1}_c$ denotes the $c$-dimensional column vector of all $1$'s.
\label{theorem potential game}
\end{theorem}
\begin{proof}
Under Assumption \ref{assump 1}, the function (\ref{eqn potential func general}) is continuously differentiable. By computing the partial derivative of (\ref{eqn potential func general}) and (\ref{eqn general cost func}) with respect to $m^i_j$, we obtain
\begin{equation*}
    2\left(\frac{\partial\Phi(u,m)}{\partial m^i_j}\right)
    =||u_i-u_j||^2_2=\frac{\partial J^i(u,m)}{\partial m^i_j}.
    \label{eqn phi derivative m^i_j}
\end{equation*}
The gradient of $\Phi(u,m)$ with respect to $u_i$ is
\begin{equation}
\begin{aligned}
     \nabla_{u_i}&\Phi(u,m)=\alpha_i\nabla_{u_i}l^i(u_i) \\
     &+\sum_{j\neq i, j\in\mathcal{N}}m^i_j(u_i-u_j)
     +\sum_{j\neq i, j\in\mathcal{N}}m^j_i(u_i-u_j).
     \label{eqn phi gradient u_i}
\end{aligned}
\end{equation}
The gradient of (\ref{eqn general cost func}) with respect to $u_i$ is
\begin{equation}
    \nabla_{u_i}\Phi(u,m)=\alpha_i\nabla_{u_i}l^i(u_i) \\
     +2\sum_{j\neq i, j\in\mathcal{N}}m^i_j(u_i-u_j).
     \label{eqn J^i gradient u_i}
\end{equation}
Under Assumption \ref{asspm 2 sysmetric network}, (\ref{eqn phi gradient u_i}) and (\ref{eqn J^i gradient u_i}) are equal. We arrive at the following equality by putting together the gradient and the partial derivatives
\begin{equation}
    \nabla_{s_i} J^i(u,m)=
    W_i \nabla_{s_i} \Phi(u,m).
    \label{eqn gradient of potential}
\end{equation}
This shows that (\ref{eqn potential func general}) is the potential function jointly minimized by all the players in $\mathcal{N}$.
Under Assumption \ref{asspm 2 sysmetric network}, players' action sets are coupled. This completes the proof.
\end{proof}
From now on we consider generalized NE defined in \cite{facchinei2010generalized}, since Assumption \ref{asspm 2 sysmetric network} couples players' action sets. 
Note that the generalized NE is a direct extension of Definition \ref{def:NE} under Assumption \ref{asspm 2 sysmetric network}.
Existence of generalized NE in $\mathcal{G}$ follows directly from Theorem 6 of \cite{facchinei2010generalized}. Following Theorem \ref{theorem potential game}, the game $\mathcal{G}$ can be considered as a constrained optimization problem where the players jointly minimize the potential function (\ref{eqn potential func general}) subject to constraints induced by Assumption \ref{asspm 2 sysmetric network}. According to \cite{monderer1996potential}, the NE of $\mathcal{G}$ can be characterized by the local minima of the potential function, since the incentives of the players to reduce their own costs are captured by the joint minimization of the potential function. This cooperative phenomenon leads to the following result.
\begin{proposition}
Under Assumption \ref{assump 1} and Assumption \ref{asspm 2 sysmetric network}, the best-response dynamics (\ref{eqn BRD}) converges. Moreover, Alg. \ref{alg:1} converges to a local minimum of (\ref{eqn potential func general}), which is a generalized NE of $\mathcal{G}$.
\label{prop brd convergence, alg. convergence}
\end{proposition}
\begin{proof}
First, we observe that when $m$ is fixed, the decision of $u$ in $\mathcal{G}$ can be considered as another weighted potential game with the action set reduced to $\mathcal{U}$. This reduced game has the same potential function as (\ref{eqn potential func general}) with only $u$ being the variable. Under Assumption \ref{assump 1}, $\min_{u\in\mathcal{U}}\Phi(u,m)$ is an optimization problem with a convex and compact feasible region and a strictly convex objective function. The strict convexity of the objective comes from the facts that the Hessian matrix of the second term in (\ref{eqn potential func general}) is positive definite and that $l^i(\cdot)$ is convex. Therefore, $\min_{u\in\mathcal{U}}\Phi(u,m)$ admits a unique optimal solution. Consider any $u_t=\begin{pmatrix}
    u_{1,t}^T & u_{2,t}^T & \dots
    & u_{N,t}^T \\
\end{pmatrix}^T\in\mathcal{U}$ which is not the global optimizer of $\min_{u\in\mathcal{U}}\Phi(u,m)$. For any $i\in\mathcal{N}$, the update given by (\ref{eqn BRD}) satisfies  $\Phi([u_{i,t+1},u_{-i,t}],m)\leq \Phi(u_t,m) , \forall i\in\mathcal{N}$. Hence, the updates  only terminate at the minimum of $\Phi(\cdot,m)$, i.e. $u_{i,t+1}=u_{i,t} ,\forall i\in\mathcal{N}$ for some $t$. This proves that (\ref{eqn BRD}) converges.
Theorem \ref{theorem potential game} also implies that $\mathcal{G}$ is also the best-reply potential. Under Assumption \ref{assump 1}, the action set $\mathcal{S}$ is convex and (\ref{eqn potential func general}) is convex separately in $u_i$ and $m^i,i\in\mathcal{N}$. According to Remark 2 of \cite{dubey2006strategic}, Alg. \ref{alg:1} converges to the NE of $\mathcal{G}$.
\end{proof}

\par
Proposition \ref{prop brd convergence, alg. convergence} shows the convenience of the potential game. Apart from the convergence results, the cooperation perspective provided by the potential game also makes the joint computation of $u$
and $m$ possible.

\begin{proposition}
Under Assumption \ref{asspm 2 sysmetric network}, the strategy profile $s^*=[u^*,m^*]$ at a generalized NE of $\mathcal{G}$ is an optimizer of the constrained optimization problem
\begin{equation}
    \begin{aligned}
        \min_{u\in\mathcal{U},m\in\mathcal{M}} & \Phi(u,m)  \\
        s.t. \  \ & m^i_j=m^j_i, \forall i,j\in\mathcal{N}, i\neq j.
    \end{aligned}
        \label{eqn jointly optimize u and m}
\end{equation}
\label{prop joint optimization of u and m}
\end{proposition}
\begin{proof}
The result follows directly from the property of the generalized weighted potential game. 
\end{proof}
Problem (\ref{eqn jointly optimize u and m}) is not convex in general. However, $\min_{u\in \mathcal{U}}\Phi(u,m) $ is a convex problem for any given $m$ as stated and $\min_{m\in\mathcal{M}}\Phi(u,m)$ subject to $ m^i_j=m^j_i, \forall i,j\in\mathcal{N}$ is a linear programming problem for any given $u$. Both of these two problems can be efficiently solved numerically. Moreover, since, in general, $\min_{u,m} \Phi(u,m)=\min_{u}\min_{m}\Phi(u,n)=\min_{m}\min_{u}\Phi(u,m)$, we can solve (\ref{eqn jointly optimize u and m}) numerically with existing optimization solvers. 
\par

\subsection{Decentralized Network Formation}
\label{sec:symmetric:decentralized}
The distributed learning under fixed networks using the best-response dynamics (\ref{eqn BRD}) is naturally decentralized. 
However, network formation using (\ref{eqn netowrk formation prob}) under Assumption \ref{asspm 2 sysmetric network} requires the joint effort of all the node to obtain an undirected network.
One observes that the equality constraints induced by Assumption \ref{asspm 2 sysmetric network} can be interpreted as the fact that the players' network formation decisions reach multiple consensuses.
Therefore, in the following, we introduce a decentralized method based on the alternating direction method of multipliers (ADMM) to obtain the network configurations inspired by \cite{boyd2011distributed}. 

Recall that the network formation problem admits the form:
\begin{equation}
\begin{aligned}
     \min_{m\in\mathcal{M}} &\Phi(u,m) \\
    \text{s.t.} \ \ &m^i_j=m^j_i, \forall i,j\in\mathcal{N}.
    \label{eq:network formation symmetric}
\end{aligned}
\end{equation}
Let $z_{ij}\in\mathbb{R}$ for all $i,,\in \mathcal{N}, j>i$ denote the auxiliary variables.
We can reformulate (\ref{eq:network formation symmetric}) using the auxiliary variables as
\begin{equation}
\begin{aligned}
     \min_{m\in\mathcal{M}} &\Phi(u,m) \\
    \text{s.t.} \ \ & m^i_j=m^j_i=z_{ij}, \forall i,j\in \mathcal{N}, j>i.
    \label{eq:network formation symmetric admm}
\end{aligned}
\end{equation}
In the sequel, we derive the ADMM updates associated with (\ref{eq:network formation symmetric admm}).
To simplify notations, we will ignore the constraint $m\in \mathcal{M}$.
The effect of this constraint is no more than involving projections on the updates.

Let $c^i_j=||u_i-u_j||^2_2$ and $c^i=(c^i_1,...,c^i_{i-1},c^i_{i+1},...,c^i_N)$.
The augmented Lagrangian of (\ref{eq:network formation symmetric admm}) with parameter $\rho>0$ is
\begin{equation}
\begin{aligned}
     \mathcal{L}_\rho:=&\frac{1}{2}\sum_{i\in\mathcal{N}}c^T_im^i  \\
     &+\sum_{i\in\mathcal{N}}\sum_{j\in\mathcal{N}, j>i}\left[ \lambda^i_j(m^i_j-z_{ij})+\lambda^j_i(m^j_i-z_{ij} ) \right]  \\
     &+\frac{\rho}{2} \sum_{i\in\mathcal{N}}\sum_{j\in\mathcal{N}, j>i}\left[ ||m^i_j-z_{ij}||^2_2+||m^j_i-z_{ij}||^2_2  \right],
     \label{eq:augmented lagrangian}
\end{aligned}
\end{equation}
where $\lambda^i_j$ denotes the dual variable associated with the constraint $m^i_j=z_{ij}$.
With a bit abuse of notations, we use $a[k]$ to denote the value of variable $a$ at iteration $k$.
Then, the ADMM updates follows from (\ref{eq:augmented lagrangian}) as:
\begin{subequations}
    \begin{equation}
    \begin{aligned}
    m^i_j[k+1]:=&\argmin_{m^i_j}  \ \ \frac{1}{2}c^i_jm^i_j+\lambda^i_j[k](m^i_j-z_{ij}[k])\\
    &+ \frac{\rho}{2}||m^i_j-z_{ij}[k]||^2_2, \forall i,j\in\mathcal{N}, i\neq j,
    \label{eq:admm 1 m}
    \end{aligned}
    \end{equation}
    \begin{equation}
    \begin{aligned}
        z_{ij}[k+1]:=& \frac{1}{2}(m^i_j[k+1] + m^j_i[k+1] )\\
        &+ \frac{1}{2\rho}(\lambda^i_j[k]+\lambda^j_i[k]), \forall i,j\in\mathcal{N}, j>i, 
        \label{eq:admm 1 z}
    \end{aligned}
    \end{equation}
    \begin{equation}
        \lambda^i_j[k+1]:=\lambda^i_j[k]+\rho(m^i_j[k+1]-z_{ij}[k+1]), \forall i,j\in\mathcal{N}, j>i, 
        \label{eq:admm 1 lambda 1}
    \end{equation}
    \begin{equation}
        \lambda^j_i[k+1]:=\lambda^j_i[k]+\rho(m^j_i[k+1]-z_{ij}[k+1]), \forall i,j\in\mathcal{N}, j>i.
        \label{eq:admm 1 lambda 2}
    \end{equation}
    \label{eq:admm 1}
\end{subequations}
Let $\Bar{m}_{ij}=\frac{1}{2}m^i_j+\frac{1}{2}m^j_i$ and $\Bar{\lambda}_{ij}=\frac{1}{2}\lambda^i_j+\frac{1}{2}$.
By substituting (\ref{eq:admm 1 z}) into (\ref{eq:admm 1 lambda 1}) and (\ref{eq:admm 1 lambda 2}), we obtain $\Bar{\lambda}_{ij}[k+1]=0, \forall i,j\in \mathcal{N}, j>i$.
This leads to $z_{ij}[k]=\Bar{m}_{ij}[k], \forall i,j\in \mathcal{N}, j>i$.
Then, (\ref{eq:admm 1}) can be formulated as the following decentralized updates:
\begin{subequations}
    \begin{equation}
    \begin{aligned}
    m^i_j[k+1]:=&\argmin_{m^i_j}  \ \ \frac{1}{2}c^i_jm^i_j+\lambda^i_j[k](m^i_j-\Bar{m}_{ij}[k])\\
    &+ \frac{\rho}{2}||m^i_j-\Bar{m}_{ij}[k]||^2_2, \forall i,j\in\mathcal{N}, i\neq j,
    \label{eq:admm 2 m}
    \end{aligned}
    \end{equation}
    \begin{equation}
        \lambda^i_j[k+1]:=\lambda^i_j[k]+\rho(m^i_j[k+1]-\Bar{m}_{ij}[k+1]), \forall i,j\in\mathcal{N}, j>i, 
        \label{eq:admm 2 lambda 1}
    \end{equation}
    \begin{equation}
        \lambda^j_i[k+1]:=\lambda^j_i[k]+\rho(m^j_i[k+1]-\Bar{m}_{ij}[k+1]), \forall i,j\in\mathcal{N}, j>i.
        \label{eq:admm 2 lambda 2}
    \end{equation}
    \label{eq:admm 2}
\end{subequations}
Note that to perform the updates in (\ref{eq:admm 2}), each node only need to collect the network formation actions of the other nodes.
Adopting (\ref{eq:admm 2}) in Algorithm \ref{alg:1} leads to a fully decentralized method to compute the NE of the game $\mathcal{G}$ when the network is undirected.
The convergence of the ADMM updates (\ref{eq:admm 2}) can be guaranteed under the assumptions of the cost functions discussed in Section \ref{sec:framework}.

\subsection{Network Structure Analysis}
\label{sec:symmetric:network structure}
With Theorem \ref{theorem potential game}, we have discovered that game $\mathcal{G}$ can be included in the class of potential games.
Hence, we can analyze the network structure under cooperative endeavors of all the players, despite the fact that the network structure is influenced by the players' local learning actions. 
In this subsection, we investigate the network structures when nodes in $\mathcal{N}$ possess different levels of budgets $\beta_i$. We focus on (\ref{eqn jointly optimize u and m}) with $u$ fixed. 
\par
Consider the case where the nodes in $\mathcal{N}$ are divided into two mutually exclusive subsets $\mathcal{N}_1$ and $\mathcal{N}_2$, i.e. $\mathcal{N}_1\cup \mathcal{N}_2=\mathcal{N}$ and $\mathcal{N}_1 \cap \mathcal{N}_2=\emptyset$. Let players in $\mathcal{N}_1$ have the same budget $a>0$ and players in $\mathcal{N}_2$ have the same budget $b>0$, and $a\neq b$. Let $\mathcal{I}_n$ denote the index set of links whose endpoints include node $n\in\mathcal{N}$, i.e. $\mathcal{I}_n:=\{(i,j)|i=n \text{ or } j=n, i\in\mathcal{N}, j\in\mathcal{N},i\neq j\}$. We reformulate the constraints on $m^i_j$ in Assumption \ref{assump 1} and Assumption \ref{asspm 2 sysmetric network} as
\begin{subequations}
    \begin{equation}
        \sum_{(i,j)\in\mathcal{I}_n}m^i_j=a, \forall n\in\mathcal{N}_1,
        \label{eqn feasibility group 1}
    \end{equation}
    \begin{equation}
         \sum_{(i,j)\in\mathcal{I}_n}m^i_j=b, \forall n\in\mathcal{N}_2.
        \label{eqn feasibility group 2}
    \end{equation}
\end{subequations}
Let $\mathcal{I}_{\mathcal{N}_1}:=\{(i,j)|i\in\mathcal{N}_1, j\in\mathcal{N}_1,i\neq j\}$ be the index set of links whose endpoints are both in $\mathcal{N}_1$. The counterpart of $\mathcal{N}_2$ is denoted by $\mathcal{I}_{\mathcal{N}_2}:=\{(i,j)|i\in\mathcal{N}_2, j\in\mathcal{N}_2,i\neq j\}$. The index set of links bridging the two subsets $\mathcal{N}_1$ and $\mathcal{N}_2$ is $\mathcal{I}_\times:=\{(i,j)|i\in\mathcal{N}_c, j\in\mathcal{N}_{3-c},\forall c\in\{1,2\} \}$. 
\par
Consider a given network $\Bar{m}^i_j,i,j\in\mathcal{N},i\neq j$, which satisfies Assumptions \ref{assump 1} and \ref{asspm 2 sysmetric network}, and hence (\ref{eqn feasibility group 1}) and (\ref{eqn feasibility group 2}). Suppose that this given network is not optimal. We use it as an initial point of the problem (\ref{eqn jointly optimize u and m}) with fixed $u$. A general way of obtaining an update of the current decision variable in an optimization problem is finding a feasible descent direction and performing line search. Let $p=[p^i_j],p^i_j \in\mathbb{R},i,j\in\mathcal{N},i\neq j$ be a feasible descent direction. By grouping the indices according to $\mathcal{N}_1$ and $\mathcal{N}_2$, we obtain the following equations:
\begin{subequations}
    \begin{equation}
        \sum_{(i,j)\in\mathcal{I}_n\cap \mathcal{I}_{\mathcal{N}_1}}p^i_j
        +\sum_{(i,j)\in\mathcal{I}_n\cap \mathcal{I}_\times} p^i_j=0, \  \ \forall n\in\mathcal{N}_1,
        \label{eqn constraints feasible descent direction gourp 1}
    \end{equation}
    \begin{equation}
         \sum_{(i,j)\in\mathcal{I}_n\cap \mathcal{I}_{\mathcal{N}_2}}p^i_j
        +\sum_{(i,j)\in\mathcal{I}_n\cap \mathcal{I}_\times} p^i_j=0, \  \ \forall n\in\mathcal{N}_2.
        \label{eqn constraints feasible descent direction gourp 2}
    \end{equation}
\end{subequations}
Summing over $n\in\mathcal{N}_1$ in (\ref{eqn constraints feasible descent direction gourp 1}) and over $n\in\mathcal{N}_2$ in (\ref{eqn constraints feasible descent direction gourp 2}) leads to
\begin{subequations}
    \begin{equation}
        \sum_{(i,j)\in \mathcal{I}_{\mathcal{N}_1}}p^i_j
        +\sum_{(i,j)\in \mathcal{I}_\times} p^i_j=0,
        \label{eqn sum feasible descent direction glabourp 1}
    \end{equation}
    \begin{equation}
         \sum_{(i,j)\in \mathcal{I}_{\mathcal{N}_2}}p^i_j
        +\sum_{(i,j)\in \mathcal{I}_\times} p^i_j=0.
        \label{eqn sum feasible descent direction gourp 2}
    \end{equation}
\end{subequations}
Equations (\ref{eqn sum feasible descent direction glabourp 1}) and (\ref{eqn sum feasible descent direction gourp 2}) can be reformulated as
\begin{equation}
     \sum_{(i,j)\in \mathcal{I}_{\mathcal{N}_1}}p^i_j
     = \sum_{(i,j)\in \mathcal{I}_{\mathcal{N}_2}}p^i_j=-\sum_{(i,j)\in \mathcal{I}_\times} p^i_j.
     \label{eqn group wise undirected descent}
\end{equation}
The interpretation of (\ref{eqn group wise undirected descent}) is that the total changes of  weights on the links within the subsets are the same, and they are opposite to the change of the total weight on the links bridging the two subsets. As a consequence, we tend to arrive at following two network structures. 
In the first structure, more links appear inside the subsets and less links appear crossing the subsets.
The second structure contains denser links between the two subsets and sparser links inside the subsets.

We remark here that the analysis above has close connection with the network cohesiveness introduced in \cite{morris2000contagion}.
Since it is challenging to derive the exact cohesiveness of a network obtained by solving (\ref{eqn jointly optimize u and m}), we have focused on the evolution of the link weights rather than the outcomes of the link weights.
Nevertheless, the patterns observed from (\ref{eqn group wise undirected descent}) uncover structural properties of the networks obtained from (\ref{eqn jointly optimize u and m}). 
We leave the explicit cohesiveness of the network to future works.

\subsection{Efficiency of Learning}
\label{sec:symmetric:efficiency of learning}
We devote this subsection to the comparison between the social welfare obtained using a standard distributed learning framework and the social welfare obtained by our game-theoretic framework. 
We will consider a fixed undirected network and focus on the efficiency of learning in the global problem captured by (\ref{eqn potential func general}).
\par
Under Assumption \ref{asspm 2 sysmetric network}, we can reformulate (\ref{eqn potential func general}) as:
\begin{equation}
    \Psi(u) =\sum_{i\in\mathcal{N}}l_i(u_i)
    +\sum_{i,j\in\mathcal{N}, j>i}m_j^i||u_i-u_j||_2^2,
\label{eqn potential func generalized}
\end{equation}
where we have absorbed the scalars $\alpha_i$ into $l_i(\cdot)$. Consider the solution to our framework given by
\begin{equation}
    P_1^*=\min_{u\in\mathcal{U}} \Psi(u),
    \label{eqn P1}
\end{equation}
and the solution to a standard distributed learning problem given by
\begin{equation}
    \begin{aligned}
         P^*_2=&\min_{u\in\mathcal{U}} \sum_{i\in\mathcal{N}}l_i(u_i) \\
        & s.t. \   \ ||u_i-u_j||_2^2=0, \forall i,j\in\mathcal{N},j>i.
        \label{eqn P2}
    \end{aligned}
\end{equation}
The constraint in (\ref{eqn P2}) is equivalent to $u_i=u_j, \forall i,j \in\mathcal{N}, j\neq i$.  The following result shows the relation between $P_1^*$ and $P_2^*$.
\begin{theorem}
Given an arbitrary undirected network, the Nash equilibrium solution (\ref{eqn P1}) lower bounds the distributed learning solution (\ref{eqn P2}), i.e. $P_1^*\leq P_2^*$.
\label{theorem P1<=P2}
\end{theorem}
\begin{proof}
Consider the Lagrangian of the constrained optimization problem (\ref{eqn P2})
\begin{equation*}
    \mathcal{L}(u,\mu)=\sum_{i\in\mathcal{N}}l_i(u_i)+\sum_{i,j\in\mathcal{N}, j>i}\mu_j^i||u_i-u_j||_2^2,
\end{equation*}
where $\mu^i_j\in\mathbb{R}, \forall i,j\in\mathcal{N},j>i$ denote the dual variables. Let $\mathcal{F}$ denote the set of $u$ that satisfies the constraints of (\ref{eqn P2}).
We obtain for any choice of $\mu^i_j\in\mathbb{R}, \forall i,j\in\mathcal{N},j>i$ and any feasible $u\in\mathcal{U}\cap \mathcal{F}$ that
\begin{equation*}
    \sum_{i\in\mathcal{N}}l_i(u_i)=\mathcal{L}(u,\mu). 
\end{equation*}
Define the dual function as $g(\mu)=\inf_{u}\mathcal{L}(u,\mu)$. For any choice of $\mu^i_j\in\mathbb{R}, \forall i,j\in\mathcal{N},j>i$ and any $u\in\mathcal{U}\cap \mathcal{F}$,
we arrive at the following inequality
\begin{equation*}
    \inf_{u\in\mathcal{U}\cap \mathcal{F}}\mathcal{L}(u,\mu)\leq\mathcal{L}(u,\mu)=\sum_{i\in\mathcal{N}}l_i(u_i).
\end{equation*}
The following inequality follows:
\begin{equation}
    \inf_{\Bar{u}\in\mathcal{U}}\mathcal{L}(\Bar{u},\mu)\leq
    \inf_{u\in\mathcal{U}\cap \mathcal{F}}\mathcal{L}(u,\mu)\leq\mathcal{L}(u,\mu)=\sum_{i\in\mathcal{N}}l_i(u_i),
    \label{eqn dualirt inequality}
\end{equation}
since $\mathcal{U}\cap \mathcal{F}$ is a subset of $\mathcal{U}$.
Let $u^*$ denote the optimizer of (\ref{eqn P2}). It is clear that  $u^*\in\mathcal{U}\cap\mathcal{F}$. Observing that (\ref{eqn dualirt inequality}) holds for any choice of dual variables $\mu^i_j\in\mathbb{R}, \forall i,j\in\mathcal{N},j>i$ and for any choice of $u\in\mathcal{U}\cap \mathcal{F}$, we pick $\mu_j^i=m_j^i,\forall i,j\in\mathcal{N}, j>i$ and $u=u^*$ in (\ref{eqn dualirt inequality}). Then, we obtain
\begin{equation}
    \inf_{\Bar{u}\in\mathcal{U}}\mathcal{L}(\Bar{u},m)\leq
    \inf_{u\in\mathcal{U}\cap \mathcal{F}}\mathcal{L}(u,m)\leq\mathcal{L}(u^*,m)=\sum_{i\in\mathcal{N}}l_i(u_i^*).
    \label{eqn dualirt inequality with m and u*}
\end{equation}
In (\ref{eqn dualirt inequality with m and u*}), the term $\inf_{\Bar{u}\in\mathcal{U}}\mathcal{L}(\Bar{u},m)$ coincides with (\ref{eqn P1}) and the term $\sum_{i\in\mathcal{N}}l_i(u_i^*)$ matches (\ref{eqn P2}). Finally, we conclude for arbitrary choice of $m$ that $P_1^*\leq P_2^*$.
\end{proof}
Theorem \ref{theorem P1<=P2} centers around the distinction between the influence of demanding a single classifier $u_i=u_j,\forall i,j\in\mathcal{N},i\neq j,$ and the effect of punishing the disagreements captured by $||u_i-u_j||_2^2,\forall i,j\in\mathcal{N},i\neq j$. 
While the disagreement term $\sum_{i\in\mathcal{N}, j>i}m_j^i||u_i-u_j||_2^2$ contributes positively to the cost function given nonnegative link weights, Theorem \ref{theorem P1<=P2} shows that the punishment induced by the disagreement is less critical than requiring $||u_i-u_j||_2^2,\forall i,j\in\mathcal{N},i\neq j$. One explanation of this phenomenon is from the saddle-point interpretation of duality \cite{boyd2004convex}. Namely, the constrained optimization problem (\ref{eqn P2}) is intrinsically a min-max problem. The dual variables $\mu^i_j, \forall i,j\in\mathcal{N},j>i$ obtained in the worst-case sense coincide with the result of maximizing (\ref{eqn potential func generalized}) over $m$.
\par
The worst-case network structure interpretation of the dual variables of (\ref{eqn P2}) is practically useful. When we use a primal-dual type of algorithm to solve (\ref{eqn P2}), the optimal dual variables indicate which network structure we should avoid. Consequently, we can refine the network structures based on past experiments using (\ref{eqn P2}).
\par
Next, we leverage Theorem \ref{theorem P1<=P2} to analyze the social welfare of learning.
Define the social welfare of learning using the classifiers $u_i,i\in\mathcal{N}$ as $\mathcal{W}(u)=-\sum_{i\in\mathcal{N}}l_i(u_i)$.
\begin{corollary}
Given any undirected network with nonnegative link weights, the social welfare of (\ref{eqn P1}) upper bounds the social welfare of (\ref{eqn P2}).
\label{coro social walfare}
\end{corollary}
\begin{proof}
Let $\Bar{u}$ denote the optimizer of (\ref{eqn P1}) and let $\Tilde{u}$ denote the optimizer of (\ref{eqn P2}).
With $m^i_j\geq 0, \forall i,j\in\mathcal{N},j>i$, we obtain $\mathcal{W}(\Bar{u})=-\sum_{i\in\mathcal{N}}l_i(\Bar{u}_i)\geq -P_1^*$. On the other hand, $\mathcal{W}(\Tilde{u})=-P_2^*$. Since $P_1^*\leq P_2^*$ according to Theorem \ref{theorem P1<=P2}, we conclude that $\mathcal{W}(\Bar{u})\geq \mathcal{W}(\Tilde{u})$. 
\end{proof}
Our framework guarantees an appealing social welfare compared to a standard distributed learning framework. In other words, our framework enables a way to find the optimal network structures which is free of concerns about the efficiency of distributed learning.

\section{Concurrent Equilibrium Seeking}
\label{sec:single loop}
Recall that under the undirected network assumption, \Cref{alg:1} converges to a generalized Nash equilibrium of $\mathcal{G}$. However, the convergence no longer holds under directed network, as $\mathcal{G}$ may not admit a weighted potential game as suggested in \Cref{theorem potential game}. In this case, the best response dynamics in \eqref{eqn BRD} may lead to a divergent learning process. For example, even in a strictly convex game where the NE is unique, it has been shown in \cite{barron10brd_fail} that best response dynamics may fail to reach the NE.

To address this convergence issue under directed network,  we propose a concurrent learning scheme based on online mirror descent (OMD) \cite{mertikopoulos2019learning}. Unlike the commutative algorithm in \Cref{alg:1}, the OMD-based one enables each player to update her learning action $u_i$ and the link weights $m^i$ simultaneously within each iteration. It will be shown later in this section that this concurrent equilibrium seeking algorithm converges to an NE of $\mathcal{G}$. 

Apart from  the convergence issue, the concurrent learning scheme considers the scenario where the inter-agent communication is subject to random noises. This suggests that a node $i\in\mathcal{N}$ may not observe the exact actions $u_{-i}$ because of the imperfect communication links. Namely, each node receives a noised feedback $\hat{u}_{-i}$ from its neighbors at each iteration. The following presents the OMD-based learning algorithm, and we begin our discussion with the distance-generating function, which is a generalization of euclidean distance (i.e., $\ell_2$-norm). 

\subsection{Multi-agent Online Mirror Descent} Online mirror descent is a class of online convex optimization techniques. The basic idea is that a new iterate is generated by taking a so-called ``mirror step'' from the last one along the direction of an ``approximate gradient'' vector, which is produced by the prox-mapping. The definition of prox-mapping relies on the idea of distance-generating function.  
\begin{definition}
For a real-valued lower-semi-continuous convex function $h$, if $\operatorname{dom}(h) = \mathcal{S}$ and 
\begin{enumerate}
\item the subdifferential of $h$ admits a continuous selection, i.e., there exists a continuous mapping $\nabla h$ such that for all $s\in \{s\in \mathcal{S}|\partial h(s)\neq \varnothing\}:=\operatorname{dom}(\partial h)$, $\nabla h(s)\in \partial h(s)$,
\item $h$ is $K-$ strongly convex, i.e., for $s\in \operatorname{dom}(\partial h)$ and $s'\in \mathcal{S}$, $$h\left(s^{\prime}\right) \geq h(s)+\left\langle\nabla h(s), s^{\prime}-s\right\rangle+\frac{K}{2}\left\|s^{\prime}-s\right\|^{2},$$
\end{enumerate}
then we say that $h$ is a distance-generating function on $\mathcal{S}$.
\end{definition}
Based on this distance-generating function, we can define a pseudo-distance or more widely known as Bregman divergence \cite{gong93bregman} via the relation
$$D(s,x)=h(s)-h(x)-\langle \nabla h(x),s-x\rangle,$$ for all $s\in \mathcal{S}$ and $s\in \operatorname{dom}(\partial h)$. Though $D$ may be directed or fail to satisfy the triangle inequality, thanks to the convexity of $h$, we can rely on $D$ to show the convergence of some sequences. Since the key property for proving the convergence of the proposed learning scheme relies on the fact that  $D(s,x)\geq \frac{1}{2}K\|x-s\|^2$, we discuss the properties of the Bregman divergence that are helpful for convergence analysis in the supplementary material. 

Finally, we arrive at the prox-mapping in the context of the game $\mathcal{G}$. Let $v_i=\nabla_{s_i}J^i(s_i,s_{-i})$ denote the payoff gradient of player $i$ at $(s_i,s_{-i})$ and $D_i(\cdot,\cdot)$ denote the Bregman divergence related to the distance-generating function $h_i$ of player $i$. The prox-mapping is given by
\begin{align}
    P_{s_i}(v_i)=\argmin_{s_i'\in \mathcal{S}_i}\{\langle v_i, s_i-s_i'\rangle+D_i(s_i',s_i)\},\text{ for }s_i\in\mathcal{S}_i.
    \label{eq:prox-map}
\end{align}

\begin{remark}[Relation to Gradient Methods]
Note that when the distance-generating function is given by $h_i(s)=\frac{1}{2}\|s\|^2$, \eqref{eq:prox-map} yields the Euclidean projection 
\begin{align}
   P_{s_i}(v_i)=\argmin_{s_i'\in \mathcal{S}_i}\{\|s_i+v_i-s_i'\|^2\}, 
\end{align}\label{eq:proj-gd}
and the resulting iterative algorithm is referred to as projected gradient ascent \cite{nesterov04book}. Hence, gradient descent (ascent) can be viewed as a special case of mirror descent. Thanks to the generic distance-generating function $h$, mirror descent allows more freedom when designing learning algorithms. For example, when $h(s)=\frac{1}{2}\|s\|_p^2, 1<p<2$, mirror descent works favorably for sparse problems \cite{lei18md-hp}. The sparsity induced by the specific distance-generating function motivates the application of mirror descent in this concurrent learning, as the sparse networks are desired under constrained communications.
\end{remark}


Under \eqref{eq:prox-map}, player $i$'s recursive scheme with variable step-size $\gamma_n$ is given by 
\begin{align}\label{eq:omd}
    s_i^{n+1}=P_{s_i^n}(-\gamma_n v_i^n),
\end{align}
where $s_i^n$ is player's current action at iteration $n$, and $v_i^n$ is the individual payoff gradient accordingly. Note that $v_i^n=\nabla_{s_i} J^i(s_i^n,s^n_{-i})=[\nabla_{u_i}J^i(s^n_i,s^n_{-i}), \nabla_{m^i}J^i(s^n_i,s^n_{-i})]$, and direct calculations give 
\begin{align*}
    &\nabla_{u_i}J^i(s^n_i,s^n_{-i})=\alpha_i \nabla l^i(u_i^n)+2\sum_{j\neq i, j\in \mathcal{N}}m_j^{i,n}(u_i^n-u_j^n),\\
    &\nabla_{m^i} J^i(s_i^n, s_{-i}^n)=[\|u_i^n-u_1^n\|^2,\ldots, \|u_i^n-u_j^n\|^2, \ldots].
\end{align*}
To compute the individual payoff gradient, each player only needs $u_j^n$ from the neighboring nodes, and \eqref{eq:omd} can be implemented independently by each player in the learning process. In the multi-agent setting, the objective function $J^i$ is jointly determined by $s_i^n$ and $s_{-i}^n$, and each player faces a moving-target problem $J^i(\cdot, s_{-i}^n)$, which is referred to as the curse of nonstationarity \cite{tao_info}. To distinguish from its single-agent counterpart where the objective function is stationary, \eqref{eq:omd} is referred to as the online mirror descent \cite{mertikopoulos2019learning}. 


Define $D(\cdot,\cdot)$ as the Bregman divergence related to the joint distance generating function $h(s)=\sum_{i\in \mathcal{N}}h_i(s_i)$.
Concatenating $s_i$ and $v_i$ using $s^n:=(s_i^n)_{i\in \mathcal{N}}$ and $v^n:=(v_i^n)_{i\in \mathcal{N}}$,
we obtain the multi-agent OMD based on (\ref{eq:omd}) as
\begin{align}\label{eq:multi_omd}
    s^{n+1}=P_{s^n}(-\gamma_n v^n),
\end{align}
where $P_s(v):=\argmin_{s'\in \mathcal{S}}\{\langle v, s-s'\rangle+ D(s',s)\}$. It is straightforward to see that the $i$-th component of $s^{n+1}$ given by (\ref{eq:multi_omd}) coincides with the one in (\ref{eq:omd}). We use this concise expression in (\ref{eq:multi_omd}) for convergence analysis.

In practice, players may receive noisy learning parameters $(\hat{s}_{j})$ from the neighbors, since the communication process may be unreliable or imperfect. To address this concern, one can consider the multi-agent OMD with stochastic gradient given by
\begin{align}\label{eq:noise_omd}
     s^{n+1}=P_{s^n}(-\gamma_n \hat{v}^n),
\end{align}
where $\hat{v}^n$ is the payoff gradient obtained using the noisy parameters $\hat{s}_n$. To streamline our convergence analysis, we focus on the deterministic case in \eqref{eq:multi_omd}, and the convergence results can be extended to the stochastic case under the martingale noise assumption \cite[section 4.1]{bravo2018bandit}.

\subsection{Convergence Analysis}
One of the promising properties of OMD is that under mild conditions, we guarantee that the sequence $\{s^n\}$ generated by (\ref{eq:noise_omd}), or equivalently the strategy profile of players in the first layer game, converges to the NE. The convergence proof relies on the following assumptions.
\begin{assumption}[Step size]\label{ass:step}
$\sum_{n=1}^\infty \gamma_n^2<\infty$, $ \sum_{n=1}^\infty \gamma_n=\infty$
\end{assumption}
The assumption is common in stochastic approximation schemes: $\sum_{n=1}^\infty \gamma_n^2<\infty$, implying $\lim_{n\rightarrow\infty}\gamma_n=0$ reduces the randomness, as learning proceeds whereas $\sum_{n=1}^\infty \gamma_n=\infty$ ensures a horizon of sufficient length.

For any joint actions $s\in \mathcal{S}$, recall that $v(s)$ is the concatenation of individual payoff gradients, i.e., $v(s)_i=\nabla_{s_i} J^i(s_i,s_{-i})$. The following assumption gives a sufficient condition for the existence of unique Nash equilibrium \cite{rosen65concave}.
\begin{assumption}[Strict Monotonicity]\label{ass:monotone}
$\langle v(s')-v(s), s'-s\rangle\geq 0, \text{ for all } s, s'\in \mathcal{S}$. The equality holds if and only if $s'=s$. 
\end{assumption}
We note that it is a technical assumption that aligns with the diagonal strict convexity (DSC) proposed in the seminal work \cite{rosen65concave}.  Much of the literature related to continuous games and applications has been built on this condition \ref{ass:monotone}. Because of the similarity between this condition and operator monotonicity conditions in optimization, games with the monotonicity condition \Cref{ass:monotone} are often referred to as monotone games. As shown in \cite{rosen65concave}, monotone games admit a unique NE, which is also the unique solution of the following variational inequality \cite{mertikopoulos2019learning}: $\sum_{i\in \mathcal{N}}\langle v_i(s), s_i-s^*_i\rangle>0, \text{ for all } s\neq s^*$.

The verification of \Cref{ass:monotone} relies a second-order test based on the Hessian block matrix of the game \cite{rosen65concave}. Denote by $H^{\mathcal{G}}=(H_{ij}^\mathcal{G})_{i,j\in \mathcal{N}}$ the Hessian matrix of the game with its $(i,j)$-block defined by
\begin{align*}
    H_{ij}^\mathcal{G}=\frac{1}{2}\frac{\partial J_i}{\partial s_j \partial s_i}+\frac{1}{2}\left(\frac{\partial J_j}{\partial s_i\partial s_j}\right)^\tp.
\end{align*}
 The second-order test proposed in \cite[Theorem 6]{rosen65concave} states that for any $s\in \mathcal{S}$, if $H^\mathcal{G}(s)$ is positive definite: $z^\tp H^\mathcal{G}(s)z>0$ for every nonzero $s$, then $\mathcal{G}$ is strictly monotone, i.e., \Cref{ass:monotone} holds. Direct calculation shows that the $H^\mathcal{G}$ is block diagonally dominant \cite{feingold65blockdiag}, as the off-diagonal matrices are singular. Hence, it is reasonable to assume \Cref{ass:monotone} holds.

To prove the convergence of a sequence $\{s_n\}$, it suffices to show that the sequence of Bregman divergence generated by the sequence converges. Based on this, we divide our convergence analysis into two parts: 1) we first show that the sequence of Bregman divergence $\{D(s^*, s_n)\}$ does converge, where $s^*$ denotes the NE, though the limit point is unknown; 2) then we show that there exists a subsequence of $\{s_n\}$ that converges to NE. The proofs of the following propositions follows the standard argument in the mirror descent literature. Due to the limit of space, we skip the proofs and refer the reader to \cite[Proposition C.2, C.3]{bravo2018bandit}  for the details.
\begin{proposition}[Convergence of Bregman divergence]\label{prop:con_breg}
Under the assumption\ref{ass:step}, \ref{ass:monotone}, for the sequence $\{s_n\}$ generated by (\ref{eq:noise_omd}), the Bregman divergence $D(s^*, s_n)$ converges to a finite random variable almost surely, where $s^*$ denotes the NE. 
\end{proposition}

\begin{proposition}[Convergence of a subsequence]\label{prop:con_sub}
Under the assumption\ref{ass:step}, \ref{ass:monotone}, there exists a subsequence $\{s_{n_k}\}$ of $\{s_n\}$, which converges to the unique NE almost surely.
\end{proposition}

Finally, we prove that the proposed multi-agent OMD converges to the NE of $\mathcal{G}$.
\begin{theorem}\label{thm:asym_con}
Under assumptions mentioned above, $\{s_n\}$, the sequence generated by (\ref{eq:noise_omd}) converges to an NE almost surely.
\end{theorem}
\begin{proof}[Proof of \cref{thm:asym_con}]
By \cref{prop:con_sub}, there is a subsequence $\{s_{n_k}\}$ such that $\lim_{k\rightarrow\infty}\|s_{n_k}-s^*\|=0$ almost surely. This implies that $\lim_{k\rightarrow\infty}D(s^*,s_{n_k})=0$ or equivalently $\liminf_{n\rightarrow \infty}D(s^*,s_n)=0$ almost surely. Then by \cref{prop:con_breg}, the limit $\lim_{n\rightarrow\infty}D(s^*,s_n)$ exists almost surely. It follows that 
\begin{align*}
    \lim_{n\rightarrow\infty}D(s^*,s_n)=\liminf_{n\rightarrow \infty}D(s^*,s_n)=0,
\end{align*}
which shows that $s_n$ converges to $s^*$ almost surely. 
\end{proof}


\section{Streaming Data}
\label{sec:streaming}
In practice, the data available for a learning task is often obtained sequentially. 
Therefore, when computation power is sufficient, one will prefer an online learning approach which allows to constantly update the learned models using new data samples. 
In this section, we adapt our framework to the scenario of streaming data by deriving a distributed Kalman filter \cite{bertsekas1997nonlinear} under quadratic cost.
In the sequel, we will adopt (\ref{eqn BRD}) in Alg. \ref{alg:1}.

\par

\subsection{Mean-Square-Error Loss}
We consider a quadratic model at each node. Define $X_i:=\sum_{k=1}^{K_i}x_{i,k}x_{i,k}^T\in\mathbb{R}^{d\times d}$ and $Y_i:=2\sum_{k=1}^{K_i}x_{i,k }y_{i,k}\in \mathbb{R}^{d\times 1}$.
Omitting the terms independent of $u_i$, we arrive at the cost function for player $i\in\mathcal{N}$ with the mean-square-error (MSE) local learning cost as
\begin{equation}
    J^i(u,m^i)=
     \frac{\alpha_i}{K_i}\left(u_i^TX_iu_i+Y_i^Tu_i\right) 
     +\sum_{j\neq i, j\in\mathcal{N}}m^i_j||u_i-u_j||^2_2. 
    \label{eqn cost func mse}
\end{equation}

In the first layer, we investigate the first-order optimality condition of (\ref{eqn cost func mse}).
Define $U_{i}:=(
    u_{1} \cdots  u_{i-1} \ \ u_{i+1}  \cdots u_{N}) \in \mathbb{R}^{d\times N-1}$. Assume that $B_i:=(\frac{\alpha_i}{K_i}X_i+\beta_i\mathrm{I})^{-1}$ exists. Using Assumption \ref{assump 1}, we obtain the best response of player $i$:
\begin{equation}
    u_{i}^*=B_i(U_{i}m^i-\frac{\alpha_i}{2K_i}Y_i).
    \label{eqn BR quadratic}
\end{equation}
To reduce the computation of inverting a matrix, when a player observes $d\geq K_i$, she can use the Sherman-Morrison-Woodbury formula \cite{woodbury1950inverting} to obtain $B_i$. 
Define $V_i=(\frac{\alpha_i}{K_i})^{\frac{1}{2}}\begin{pmatrix}
x_{i,1} & x_{i,2} & \dots & x_{i,K_i} \\
    \end{pmatrix} \in \mathbb{R}^{d \times K_i}$.
Then, we obtain that the $d\times d$ matrix admits $B_i=\beta_i^{-1}\mathrm{I}-V_i[\mathrm{I}+\beta_iV_i^TV_i]^{-1}V_i^T$, which only involves the inversion of the $K_i\times K_i$ matrix $\mathrm{I}+\beta_iV_i^TV_i$. The best-response dynamics (\ref{eqn BRD}) directly follow from (\ref{eqn BR quadratic}).

\par
In the second layer, we obtain the objective of (\ref{eqn netowrk formation prob}) by substituting (\ref{eqn BR quadratic}) into (\ref{eqn cost func mse}). Defining $A_i:=U_im^i-\frac{\alpha_i}{2K_i}Y_i$, we obtain the objective as follows:
\begin{equation}
    \begin{aligned}
        &J^i(u^*_i,u^*_{-i},m^i)= 
        \frac{\alpha_i}{K_i}A_i^TB_i^TX_iB_i A_i  +\mathrm{1}^Tm^iA_i^TB_i^TB_iA_i 
      \\
        &  + \frac{\alpha_i}{K_i}Y_i^TA_i 
        -2\sum_{j\neq i,j\in\mathcal{N}}m^i_ju_j^TB_iA_i 
        + 
        \sum_{j\neq i,j\in\mathcal{N}}m^i_ju_j^Tu_j.
    \end{aligned}
    \label{eqn objective of network formation}
\end{equation}

\subsection{Streaming Data}

Consider the scenario where player $i$ observes a new data point $\{x_{i,k},y_{i,k}\}$ at each of the arrival times $k=1,2,...,K_i$. Let $u_i^k$ be the learning decision of player $i$ chosen based on data $\{x_{i,k'},y_{i,k'}\}_{k'=1}^{k}$. In a fixed iteration of the first layer of Alg. \ref{alg:1}, player $i$ observes the current network structure $m^{i}$ and the current learning actions of connected players $u_j$ such that $m^i_j\neq 0$. Define $X_i^k:=\frac{1}{k}\sum_{k'=1}^{k}x_{i,k'}x_{i,k'}^T$ and $Y_i^k:=\frac{2}{k}\sum_{k'=1}^{k}x_{i,k'}y_{i,k'}$. Player $i$'s cost function at time $k$ is
\begin{equation}
    J^{'i}_k(u_i^k)=\alpha_i\left( u_i^{kT}X_i^ku_i^k+Y_i^{kT}u_i^k \right)+
    \sum_{j\neq i,j\in\mathcal{N}}m^{ik}_j||u_i^k-u_j||^2_2.
    \label{eqn kalman cost}
\end{equation}
Define $\Lambda_i^k:=\alpha_iX_i^k+(\mathrm{1}^Tm^{i})\mathrm{I}$ and $\Gamma_i^{kT}:=\alpha_i Y_i^{kT}-2\sum_{j\neq i, j\in\mathcal{N}}m^{i}_j u_j^{T}$, where $\mathrm{1}$ is the vector of all ones and $\mathrm{I}$ is the identity matrix. We drop the terms independent of $u_i^k$ and reformulate (\ref{eqn kalman cost}) as 
\begin{equation*}
    J^{'i}_k(u_i^k)=u_i^{kT}\Lambda_i^ku_i^{k}+\Gamma_i^{kT}u_i^k.
\end{equation*}
The choice of learning decision $u_i^{k}$ at $k$ solves $\nabla_{u_i^k}J^{'i}_k(u_i^k)=0$, which is equivalent to 
\begin{equation*}
    2\Lambda_i^k u_i^k+\Gamma_i^k=0.
    \label{eqn FOC kalman}
\end{equation*}
At $k+1$, we observe a new data point $\{x_{i,k+1},y_{i,k+1}\}$ and update the previous action $u_i^k$ based on it to obtain $u_i^{k+1}$.

Define $\Sigma_i^k:=\frac{1}{k+1}\mathrm{1}^T m^{i}\mathrm{I}+\frac{\alpha_i}{k+1}x_{i,k+1}x_{i,k+1}^T$ and $\Omega_i^k:=\frac{\alpha_i}{k+1}x_{i,k+1}^Ty_{i,k+1}-\frac{1}{k+1}\sum_{j\neq i}m^{i}_j u_j^{T}$ for all $i\in\mathcal{N}$.
Player $i$'s update of leaning action at each arrival instance $k=1,2,...,K_i-1$ is given recursively by
\begin{equation}
    u_i^{k+1}=u_i^k-(\frac{k}{k+1}\Lambda_i^k+\Sigma_i^k)^{-1}\Sigma_i^ku_i^k-\Omega_i^k,
    \label{eqn Kalman filter}
\end{equation}
starting from $u_i^1=-\frac{1}{2}(\Lambda_i^1)^{-1}\Gamma_i^1$.
\par
From the definitions of $\Sigma_i^k$ and $\Omega_i^k$, we observe that there is no transmission of data among nodes. The recursive structure of (\ref{eqn Kalman filter}) enables the incremental updates of learning actions based on the current learning action and the new data point at all the nodes in a distributed fashion. The above derivations assume that the other players' learning actions $u_j^k, \forall j\neq i$ are fixed. We can also consider communications of learning actions at each arrival time of local data point. In this case, we obtain an update rule similar as (\ref{eqn Kalman filter}), with $u_j^k, k=1,2,...,K_j, j\in\mathcal{N}$ substituting $u_j,j\in\mathcal{N}$. By exchanging learning actions at each arrival $k$, we constantly synchronize new information in the new data points observed by all the nodes while maintain the distributed architecture where data stays local.
\par
We remark that the generalization of (\ref{eqn Kalman filter}) when the players' cost functions are of the form of (\ref{eqn general cost func}) leads to the extended Kalman filter \cite{bertsekas1997nonlinear}.

\section{Case Study}
\label{sec:case}
In this section, we elaborate the commutative framework under (\ref{eqn BRD}) and corroborate the results numerically using data of telemonitoring measurements of Parkinson's disease \cite{tsanas2009accurate}. 
\par
We consider a scenario where hospitals, represented by nodes aim to create a network to exchange information to improve the network-wide learning in order to improve services. Some of the hospitals are general hospitals while the others  dedicate to specific medical specialties. As a consequence, local data collected from past patients at each hospital varies. The hospitals use machine learning on the local data to provide medical services. Since the local data can be inadequate, a hospital may want to seek inputs from other hospitals. We use our framework to enable collaborations among the hospitals so that the learning results from the connected hospitals improve its local services. 
In addition, every hospital is constantly collecting data from patients and updating its existing services. We address streaming data at all the hospitals in one iteration of the first layer of our game-theoretic model.

To model general hospitals, we create unbiased data from the telemonitoring dataset by performing a random shuffling at the beginning of each experiment. To model specialty hospitals, we sort the measures (features) according to the scores (labels). After the sorting, the closer the nodes' indices, the more similar the local data at these nodes.

\par
We first examine network structures under Assumption \ref{asspm 2 sysmetric network}.
Consider $N=6$ and $\alpha_i=1, \forall i\in\mathcal{N}$.
\begin{figure}[ht]
\centering
\vspace{-5mm}\subfigure[$\beta_i=1, \forall i\in\mathcal{N}$]{\includegraphics[width=0.22\textwidth]{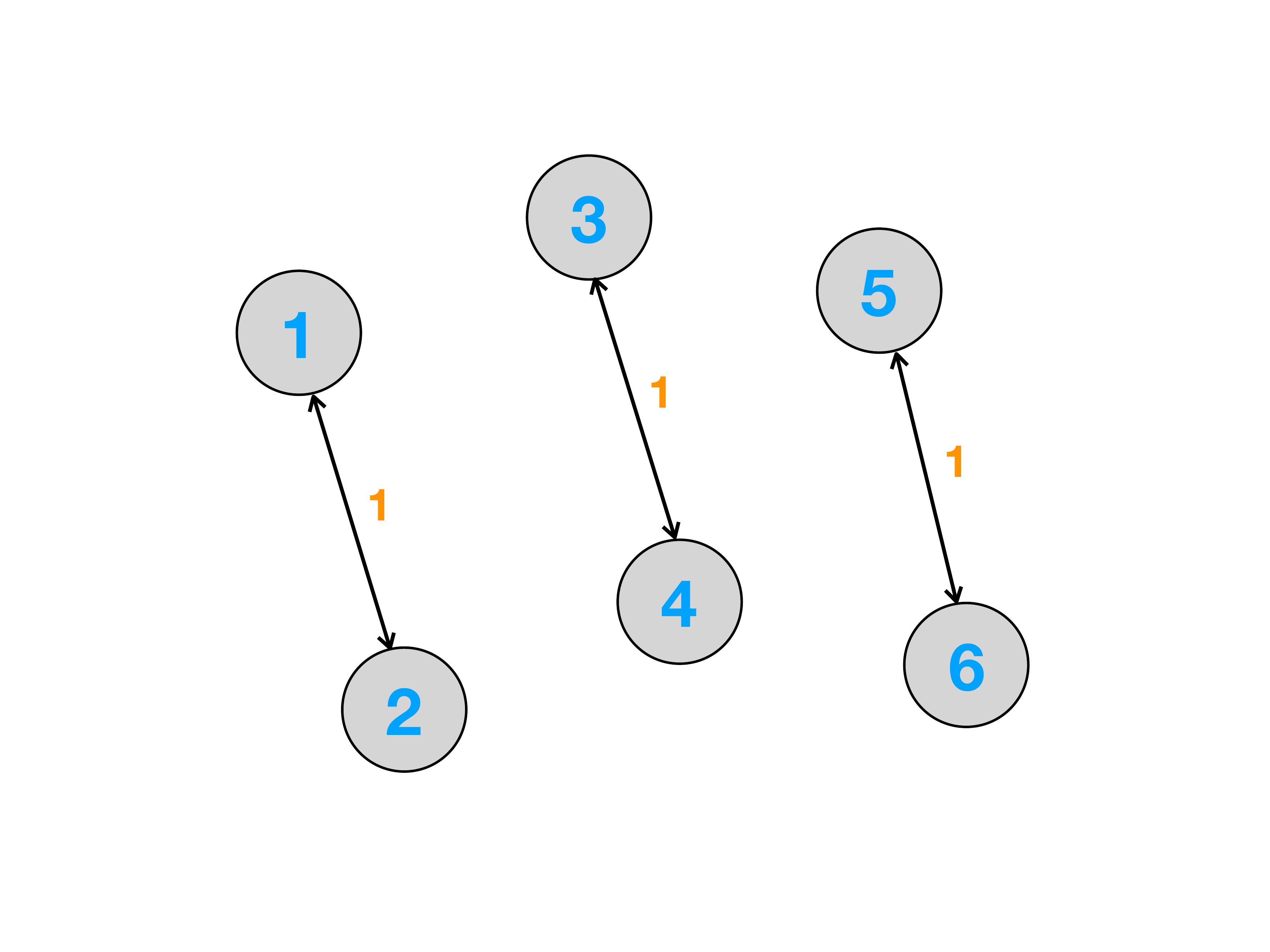}
    \label{fig:subfig1}
}
\subfigure[$\beta_i=1.5, \forall i\in\mathcal{N}$]{\includegraphics[width=0.22\textwidth]{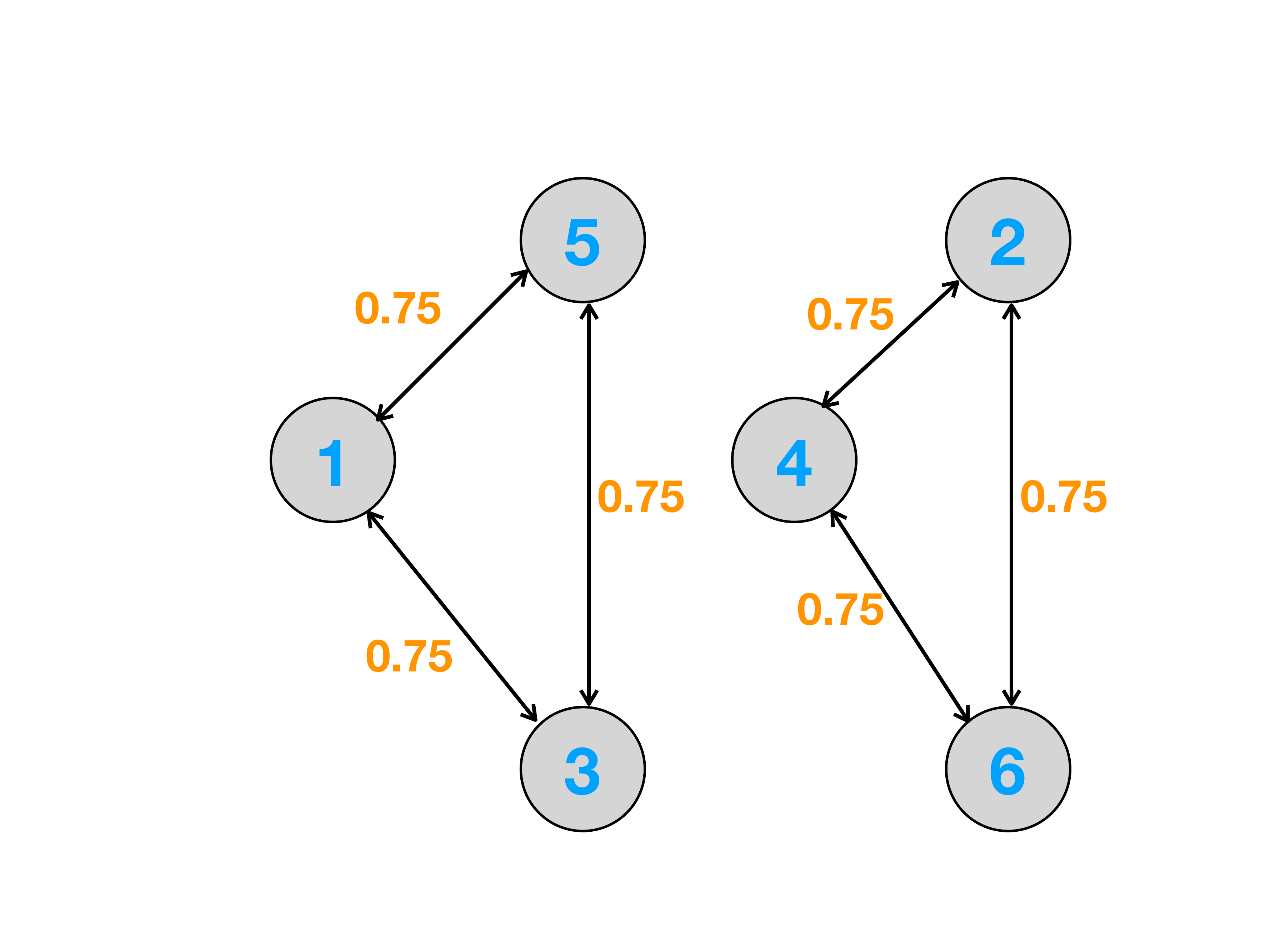}
    \label{fig:subfig2}
}
\caption{Networks structures under even budgets and unbiased data. The networks are not connected. Links of the sub-networks have the same weights.} 
\label{fig:fig1}
\end{figure}

\begin{figure}[ht]
\centering
\vspace{-10mm}\subfigure[Dividing $\mathcal{N}$ into $2$ subsets.]{\includegraphics[width=0.22\textwidth]{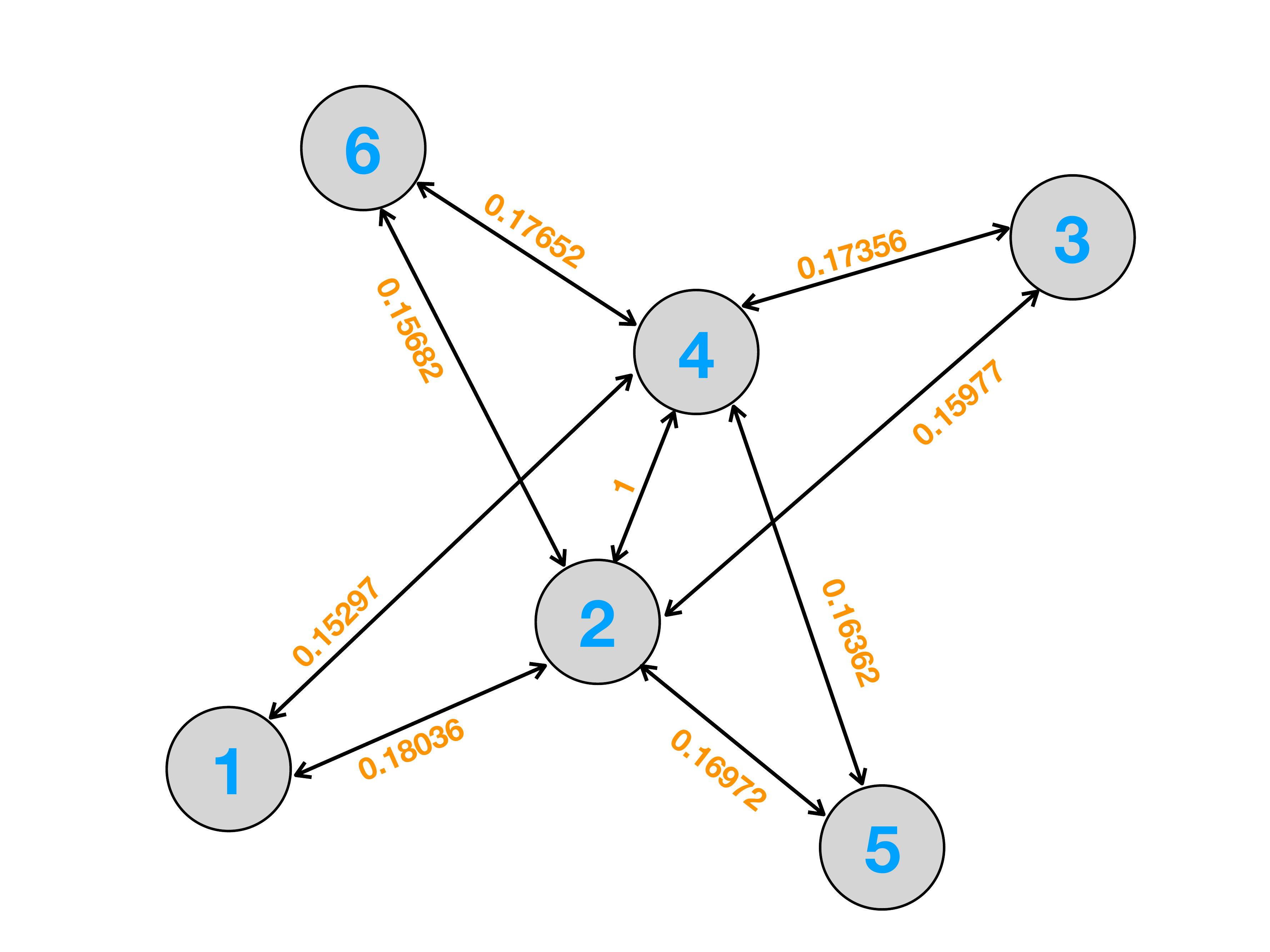}
    \label{fig:subfig3}
}
\subfigure[Dividing $\mathcal{N}$ into $3$ subsets.]{\includegraphics[width=0.22\textwidth]{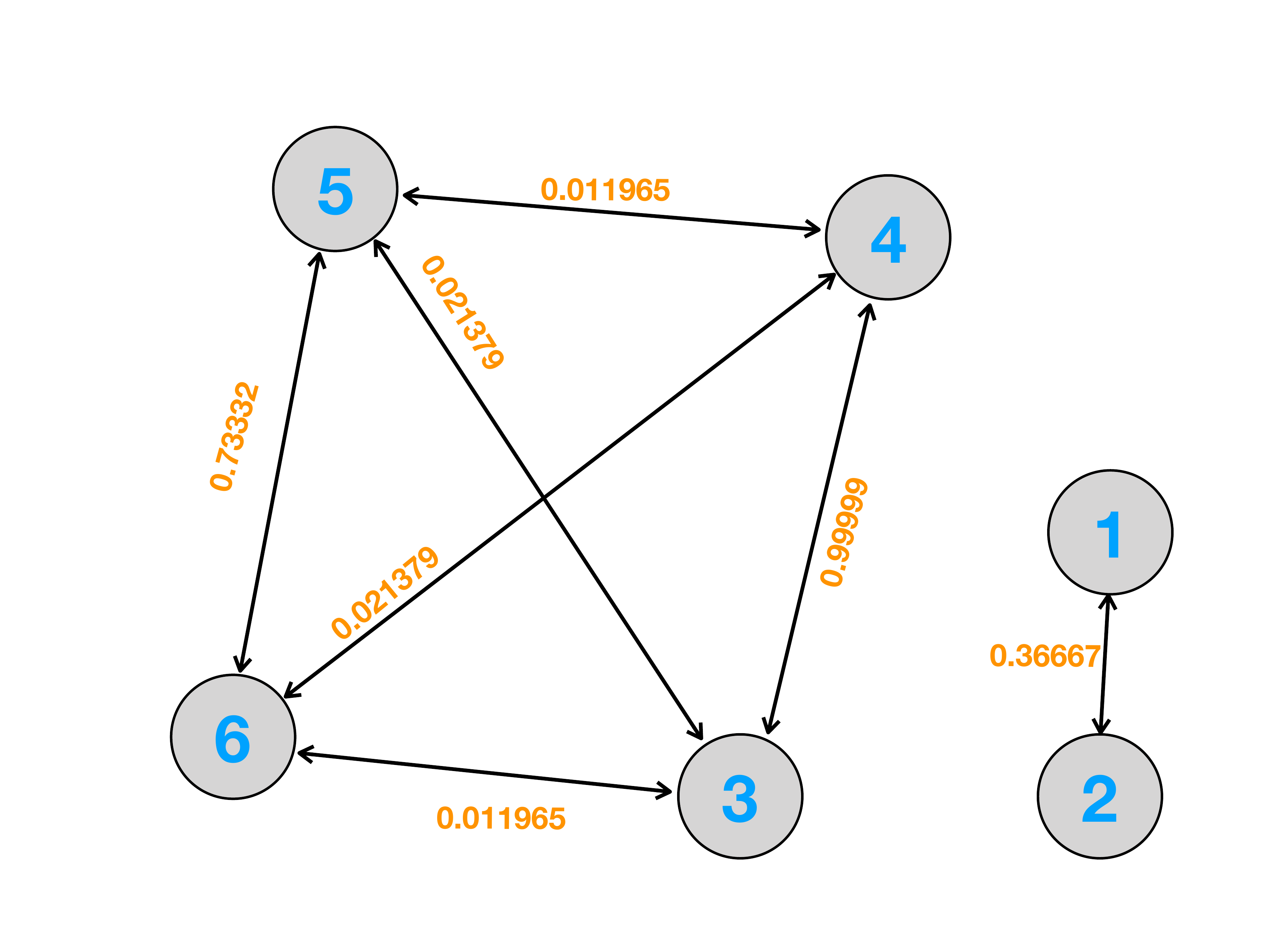}
    \label{fig:subfig5}
}
\caption{Networks structures when $\mathcal{N}$ is divided into subsets based on nodes' budgets. The links in (a) concentrate on bridging nodes from different subsets. The links in (b) concentrate on bridging nodes from the same subset. } 
\label{fig:fig2}
\end{figure}

\subsection{Structures of Small Networks}
\label{sec:case:small nets}
Network structures obtained using our framework with unbiased data are shown in Fig. \ref{fig:fig1} and Fig. \ref{fig:fig2}. They are sparse networks and corroborate the discussions around (\ref{eq:equivalent player i's network formation problem}).  Fig. \ref{fig:fig1}
shows that the most efficient network may not be connected. 
Based on Corollary \ref{coro social walfare}, the disconnectedness implies that a fully connected network may not be optimal in terms of social welfare from a macroscopic perspective. In practice, this suggests dividing the nodes into subsets and considering games of smaller scales, where one node's learning parameters only depend on nodes whose local data serve as a useful reference. 
Equivalently, these games of smaller scales correspond to a disconnected information structure, where nodes included in a game share similar data or model for distributed learning.
Furthermore, this disconnectedness shows that the condition that all the nodes reach the same learning parameter as required by a standard distributed learning framework may be overly demanding in view of the universal learning loss. 
\par
Results in Fig. \ref{fig:fig2} justify the discussions in Section \ref{sec:symmetric:network structure}. When the players with the same budget are grouped together by the subsets of $\mathcal{N}$, the link weights tend to gather either inside individual subsets or between different subsets. 
The structure of Fig. \ref{fig:fig2} (a) coincides with the core-periphery architectures discussed in \cite{galeotti2010law}. In our setting, a larger curiosity value $\beta_i$ indicates a lower quality of local data. This observation seems opposite to the one in \cite{galeotti2010law}, where the core nodes have a higher information level. However, under the potential game, nodes in our framework can be viewed to choose the network structure cooperatively. Nodes with low-quality data have a higher potential in improving the overall performance. Hence, the joint network formation makes these nodes at the core positions. From the perspective of Section \ref{sec:symmetric:network structure}, the core-periphery architecture is the consequence of link weights gathering between the subsets of $\mathcal{N}$. Nodes at the core positions are from those subsets of $\mathcal{N}$ which possesses small cardinalities and large budgets.
We also perform experiments when the data is biased and observe similar patterns.

\par
\begin{figure}[ht]
\centering
\subfigure[Unbiased data. Link weights of all the nodes concentrate on a subset of nodes. ]{\includegraphics[width=0.45\textwidth]{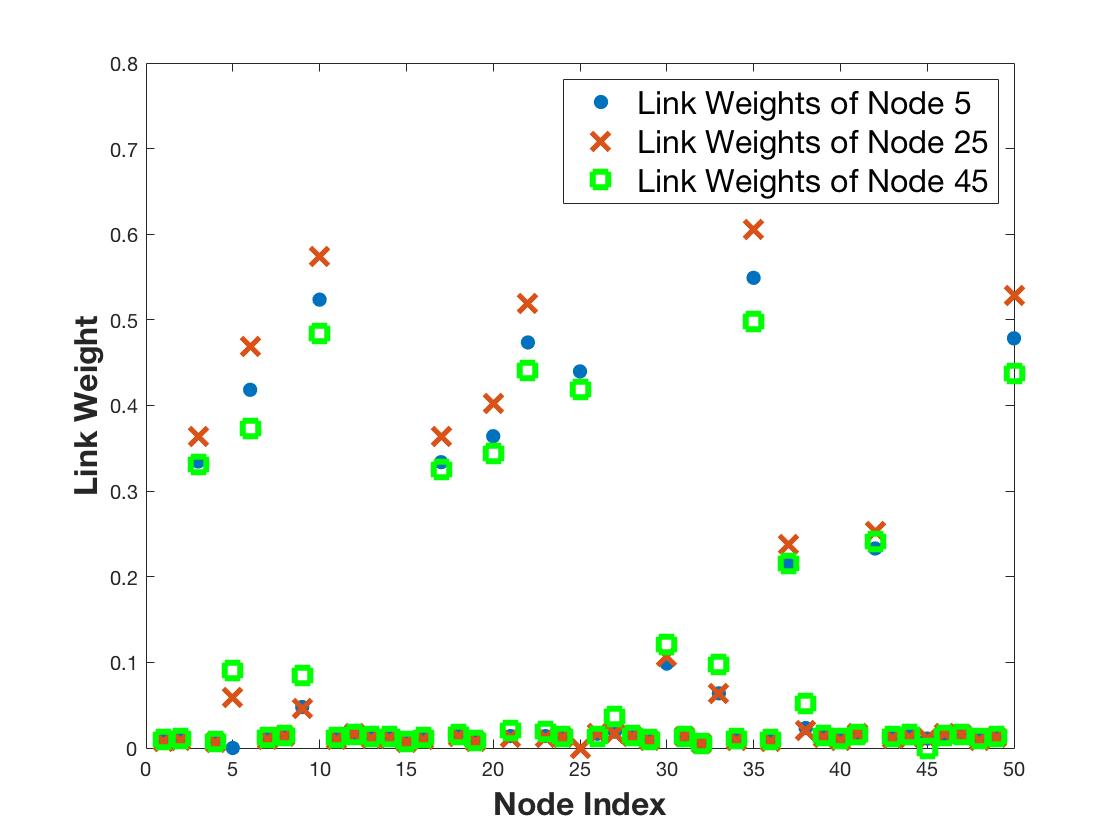}
    \label{fig:subfig5}
}

\subfigure[Biased data. Link weights of a node concentrate on the nodes whose indices are close to it. ]{\includegraphics[width=0.45\textwidth]{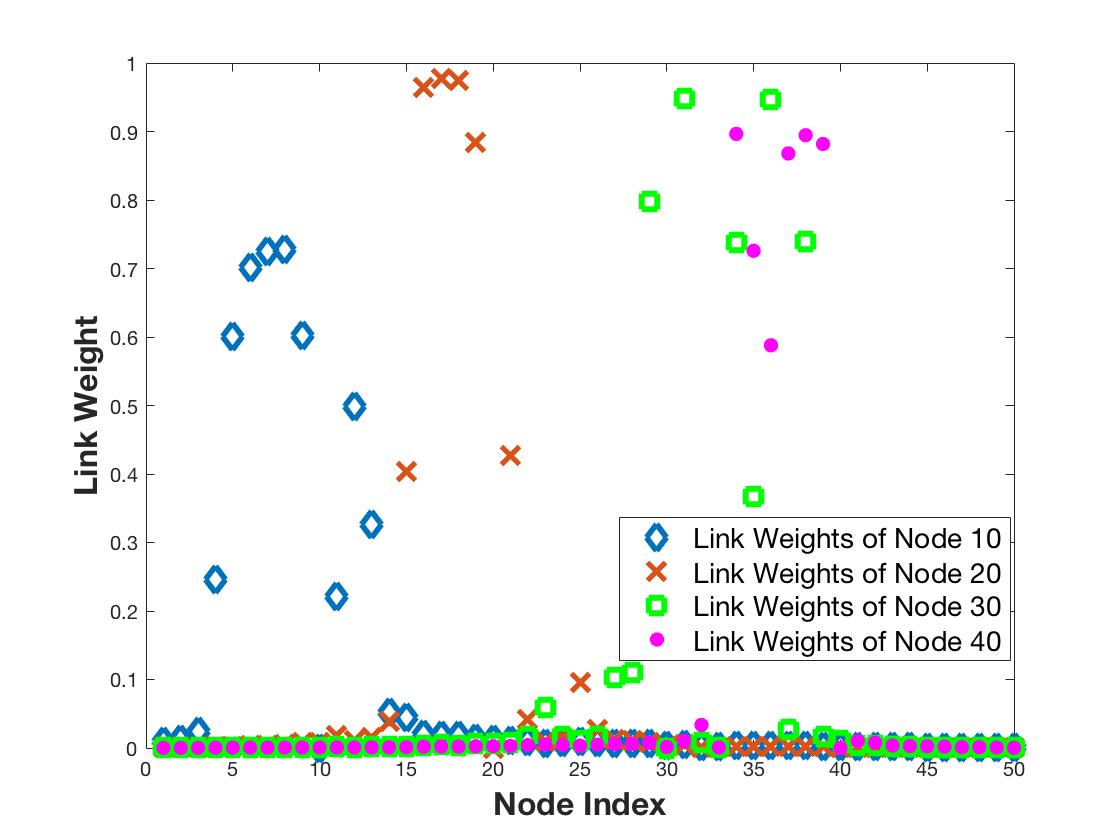}
    \label{fig:subfig6}
}

\caption[Optional caption for list of figures]{Link weight distributions of nodes. The total number of nodes $N$ is $50$. All the nodes $i\in\mathcal{N}$ has the same budget $\beta_i=5$. } 
\label{fig:3}
\end{figure}

\subsection{Link Weight Distributions of Large Networks}
\label{sec:case:large nets}
When the number of nodes becomes large, we present the link weight distributions obtained using Alg. \ref{alg:1} without Assumption \ref{asspm 2 sysmetric network} in Fig. \ref{fig:3}. The results of unbiased data and biased data have different patterns. As in Fig. \ref{fig:3} (a), when the data is unbiased, link weight distributions of all nodes in $\mathcal{N}$ are approximately the same. A subset of nodes is more popular than the others since all nodes connect with them through large link weights. These nodes represent the hospitals whose local data is either credible or abundant. By connecting with them, a hospital increases its own learning results and thus improves local medical services. 
When the data are biased, the link weights of a node concentrate on the nodes whose local data are similar to its own as in Fig. \ref{fig:3} (b). In this case, the connected nodes possess data of similar usage, which shows the ability of our framework in selecting useful targets to connect with. A node does not connect with another node whose index stays too far. The reason lies in that the distributions of data at these two nodes vary too much. Hence, a connection between these nodes will not improve local learning performances. 

\subsection{Effects of Reference Information}
\label{sec:case:connection}
Our game-theoretic framework provides extensibility to the network, making it change dynamically according to the learning parameters of nodes. In Alg. \ref{alg:1}, each refinement of the network structure in the second layer can involve 
arrival or departure of nodes. This dynamical environment is out of nodes' own willingness since 
nodes build connections by themselves. On the one hand, an existing hospital network will not block a new hospital whose local learning results are inspiring to the community. On the other hand, a hospital can disconnect from a network at any time if it benefits the hospital's own services. 
\par
\begin{figure}[ht]
\centering
\includegraphics[width=0.5\textwidth]{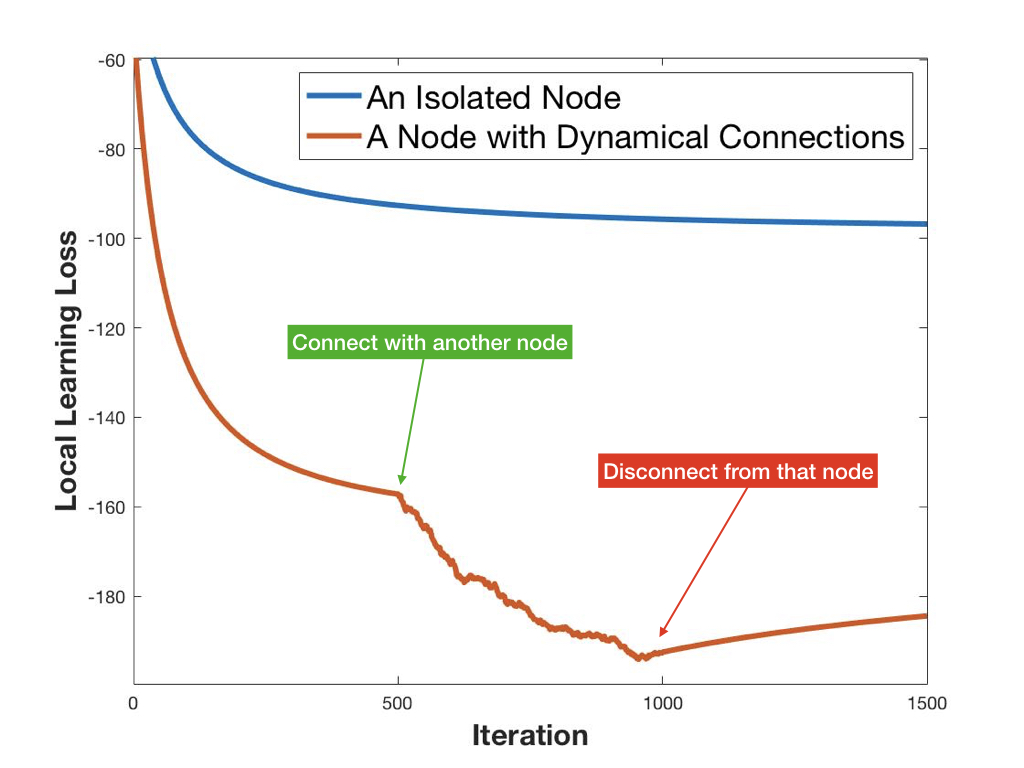}
\caption{Comparison of local learning losses between an isolated node and a node with dynamical connections. The connection improves local learning. } 
\label{fig:loss with arrival and leave}
\end{figure} 
Fig. \ref{fig:loss with arrival and leave} shows a beneficial connection. With a reference learning result, a hospital observes an improvement in its local medical services. When the reference becomes unavailable, the medical services of the hospital gradually return to the state where the connection is absent.
At this state, the learning result of the hospital emphasizes the welfare of local patients. When the node connects with others at iterations 500 to 1000 in Fig. \ref{fig:loss with arrival and leave}, she observes noise coming from the communication link. Our learning algorithm based on OMD guarantees a stochastic convergence. Hence, a hospital in the network can focus on improving medical services rather than technical details related to communication links. 
\par
We have also performed similar experiments on unbiased data but did not observe obvious improvements of local learning losses when a node connects with others. The reason lies in that local data is sufficiently ample and credible after the shuffling.

\par

\section{Conclusion}
\label{sec:conclusion}
In this paper, we have introduced a game-theoretic framework for distributed learning over networks. 
In the framework, we have modeled the link weights on a network as rational choices of players.
In addition to the learning parameters, this design of player's actions has made the configurations of networks outputs from the game.
We have presented a commutative method to obtain equilibrium learning and network formation actions of players. We have observed that the game belongs to the class of potential games when we consider undirected networks. 
Leveraging this fact, we have shown the convergence of the proposed commutative algorithm. 
Besides, we have analyzed the properties of the network structures that are likely to appear. Furthermore, we have proved that our framework performs no worse than standard distributed learning frameworks in the sense of social welfare. 
A concurrent method has also been proposed to solve the game in a generic setting.
Furthermore, we have adapted our framework to scenarios involving streaming data by deriving a distributed Kalman filter.
In the numerical experiments, we have shown possible network patterns obtained from our framework. We have illustrated the extensibility provided by our framework by showing the change of the local learning loss at a node when there are new connections with or disconnections from other nodes. 

\par
Our framework has presented a general class of machine learning algorithms. In future work, we would characterize the connectivity of the resulting networks. Partitions of the nodes would provide insights on the relations between local data at different nodes. We would also extend the game settings. A Stackelberg game would capture behaviors of nodes when a monopoly of data exists, and a Bayesian game would help us understand the incentives of nodes to reveal fake learning information when data security is a concern.

\ifCLASSOPTIONcaptionsoff
  \newpage
\fi



%



\bibliographystyle{IEEEtran}
\bibliography{bibliography.bib}
\nocite{*}

\end{document}